\documentclass[reqno]{amsart}
\usepackage{amsmath,amsthm,amsfonts,amssymb,srcltx}
\usepackage{graphicx} 
\usepackage{epstopdf}
\usepackage{enumerate}
\usepackage{simplewick}

\usepackage{color}

\newcommand{\bea}{\begin{eqnarray}}
\newcommand{\eea}{\end{eqnarray}}
\newcommand{\be}{\begin{equation}}
\newcommand{\ee}{\end{equation}}

\def\IZ{\mathbb {Z}}

\def\IR{\mathbb {R}}

\newtheorem{theorem}{Theorem}[section]
\newtheorem{proposition}[theorem]{Proposition}
\newtheorem{lemma}[theorem]{Lemma}
\newtheorem{corollary}[theorem]{Corollary}
\theoremstyle{definition}
\newtheorem{definition}[theorem]{Definition}

\newtheorem{remark}[theorem]{Remark}

\renewenvironment{proof}{{\noindent\bf Proof.}}{\hfill $\Box$\par\vskip3mm}

\makeatletter
 
 \@addtoreset{equation}{section}
\makeatother

\newcommand{\ii} {\pmb{i}}

\begin{document}


\author[J.E. Andersen]{J{\o}rgen Ellegaard Andersen}
\address{QGM\\
Department of Mathematics\\
Aarhus University\\
DK-8000 Aarhus C\\
Denmark}
\email{jea.qgm{\char'100}gmail.com}

\author[H. Fuji]{Hiroyuki Fuji}
\address{Faculty of Education\\ 
Kagawa University\\
Takamatsu 760-8522\\
Japan;
QGM\\
Aarhus University\\
DK-8000 Aarhus C\\
Denmark
}
\email{fuji{\char'100}ed.kagawa-u.ac.jp}

\author[M. Manabe]{Masahide Manabe}
\address{Faculty of Physics\\
University of Warsaw\\
ul. Pasteura 5, 02-093 Warsaw\\
Poland}
\email{masahidemanabe{\char'100}gmail.com}

\author[R. C. Penner]{Robert C. Penner}
\address{Institut des Hautes {\'E}tudes Scientifiques, 35 route de Chartres, 91440 Burs-sur-Yvette, France;
Division of Physics, Mathematics and Astronomy, California Institute of Technology, Pasadena, CA 91125, USA}
\email{rpenner{\char'100}caltech.edu,\hspace{0.3cm}rpenner@ihes.fr}

\author[P. Su{\l}kowski]{Piotr Su{\l}kowski}
\address{Faculty of Physics, University of Warsaw, ul. Pasteura 5, 02-093 Warsaw, Poland; Walter Burke Institute for Theoretical Physics, California Institute of Technology, Pasadena, CA 91125, USA}
\email{psulkows{\char'100}fuw.edu.pl}

\title[Partial chord diagrams and matrix models]{Partial chord diagrams and matrix models}

\thanks{{\bf Acknowledgments:}
JEA and RCP is supported in part by the center of excellence grant ``Center for Quantum Geometry of Moduli Spaces'' from the Danish National Research Foundation (DNRF95).
The research of HF is supported by the
Grant-in-Aid for Research Activity Start-up [\# 15H06453], Grant-in-Aid
for Scientific Research(C)  [\# 26400079], and Grant-in-Aid for Scientific
Research(B)  [\# 16H03927]  from the Japan Ministry of Education, Culture,
Sports, Science and Technology.
The work of MM and PS is supported by the ERC Starting Grant no. 335739 ``Quantum fields and knot homologies'' funded by the European Research Council under the European Union's Seventh Framework Programme.
PS also acknowledges the support of the Foundation for Polish Science, and RCP acknowledges the kind support of
Institut Henri Poincar\'e where parts of this manuscript were written.}
\begin{abstract}
In this article, the enumeration of partial chord diagrams is discussed via matrix model techniques.
In addition to the basic data such as the number of backbones and chords, we also consider 
the Euler characteristic, the backbone spectrum, the boundary point spectrum, and the boundary length spectrum.
Furthermore, we consider the boundary length and point spectrum that unifies the last two types of spectra. 
We introduce matrix models that encode generating functions of partial chord diagrams filtered by each of these spectra.
Using these matrix models, we derive partial differential equations -- obtained independently by cut-and-join arguments in an earlier work -- for the corresponding generating functions.
\end{abstract}

\maketitle
\tableofcontents

\section{\label{sec:1}Introduction}

A {\em partial chord diagram} is a special kind of graph, which is specified as follows. The graph consists of a number of line segments (which are called {\em backbones}) arranged along the real line (hence they come with an ordering), with a number of vertices on each. A number of semi-circles (called {\em chords}) arranged in the upper half plane is attached at a subset of the vertices of the line segments, in such a way that no two chords have endpoints at the same vertex. The vertices which are not attached to chord ends are called the marked points. A {\em chord diagram} is by definition a partial chord diagram with no marked points. Partial chord diagrams occur in many branches of mathematics, 
including topology \cite{Ba,Kontsevich2}, geometry \cite{AMR1,AMR2,ABMP}
and representation theory \cite{C-SM}. 

To each partial chord diagram $c$ one can associate canonically a two dimensional surface with boundary $\Sigma_c$, see Figure \ref{partial_chord}.
Moreover, as discussed in \cite{VOZ,APRWat, AAPZ, Andersen_new}, 
the notion of a {\em fatgraph} \cite{P1,P2,P3,P4} is a useful concept when studying partial chord diagrams.
A fatgraph is a graph together with a cyclic ordering on each collection of half-edges 
incident on a common vertex.
A partial linear chord diagram $c$ has a natural fatgraph structure induced from its presentation in the plane.

\begin{figure}[h]
\begin{center}
   \includegraphics[width=120mm,clip]{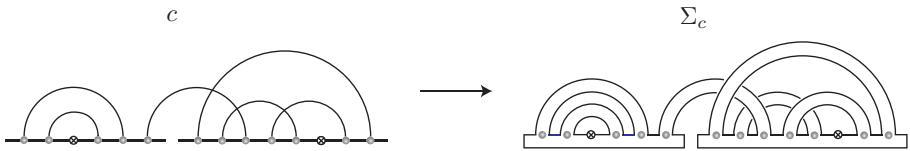}
\end{center}
\caption{\label{partial_chord} The partial chord diagram (with marked points) $c$ and the corresponding surface $\Sigma_c$.
 The type of this  partial chord diagram reads
 $\{g,k,l;\{b_i\};\{n_i\};\{p_i\}\}=\{1,6,2;\{b_6=1,b_8=1\};\{n_0=2,n_1=2\};\{p_1=1,p_2=2,p_9=1\}\}$. The boundary length and point spectrum is $\{n_{(1)}=1,n_{(0,0)}=2,n_{(0,0,0,0,0,1,0,0,0)}=1\}$.
 }
\end{figure}

The partial chord diagram $c$ is characterized by various topological data, and we will consider the following five types of data, introduced in \cite{AAPZ} and \cite{Andersen_new}.
\begin{itemize}

\item The number of chords $k$ in $c$ and the number  of backbones $b$ in $c$.

\item Euler characteristic $\chi$ and genus $g$.\\
Let $\chi$ and $g$ denote respectively the Euler characteristic and genus of $\Sigma_c$, which are related as follows
\begin{align}
\chi=2-2g.
\nonumber
\end{align}
Denoting by $n$ the number of boundary components of $\Sigma_c$,  the Euler relation can be written as
\begin{align} 
2-2g=b-k+n.
\end{align}

\item Backbone spectrum $(b_0,b_1,\ldots)$.\\
Let $b_i$ denote the number of backbones with $i$ trivalent (i.e. chord ends) or bivalent (i.e. marked points) vertices.
The total number of backbones $b$  is then
\begin{align}
b=\sum_{i\ge 0}b_i,
\end{align}
and the total number $m$ of trivalent (i.e. chord ends) and bivalent (i.e. marked points) vertices 
of the partial chord diagram $c$ is
\begin{align}
m=\sum_{i\ge 1}ib_i.
\end{align}

\item Boundary point spectrum $(n_0,n_1,\ldots)$.\\
Let $n_i$ denote the number of boundary components containing $i\ge 0$ 
marked points of $\Sigma_c$.
The total number $n$ of boundary components 
is 
\begin{align}
n=\sum_{i\ge 0}n_i,
\end{align}
and the total number $l$ of marked points is 
\begin{align}
l=\sum_{i\ge 1}i n_i.
\end{align}
These three numbers $m$, $k$ and $l$ satisfies
\begin{align}
m=2k+l .
\end{align}

\item Boundary length spectrum $(p_1,p_2,\ldots)$.\\
Define the {\em length} of a boundary component to be the sum of the number of 
 chords and the number of backbone undersides traversed by the boundary cycle. 
Let $p_i$ be the number of boundary cycles with length $i\ge 1$.
By definition,  the following two relations hold
\begin{align}
&n=\sum_{i\ge 1}p_i,\\
&2k+b=\sum_{i\ge 1}ip_i.
\end{align}

\end{itemize}
The data $\{g,k,l;\{b_i\};\{n_i\};\{p_i\}\}$  
is called the {\em type} of a partial chord diagram $c$.

As a unification of the boundary length spectrum and the boundary point spectrum, we consider the {\em boundary length and point spectrum} introduced in \cite{Andersen_new}. Let us here recall its definition.
\begin{itemize}
\item Boundary length and point spectrum.\\
We associate a $K$-tuple of numbers $\ii=(i_1,\ldots,i_K)$ with a boundary component of length $K$, where $i_L$ ($L=1,\ldots,K$) is the number of marked points between the $L$'th and $(L+1)$'th (taken modulo $K$) either chord or underpass of a backbone component (in either order) along the boundary. 

Let $n_{\ii}$ be the number of boundary components labeled in this way by $\ii$.
The total number $l$ of marked points is
\begin{align}
&l=\sum_{K\ge 1}\sum_{\ii}\sum_{L=1}^Ki_Ln_{(i_1,\ldots,i_K)},
\end{align}
and the total number $n$ of boundary cycles is
\begin{align}
n=\sum_{\ii}n_{\ii}.
\end{align}
\end{itemize}
The data $\{g,k,l;\{b_i\},\{n_{\ii}\}\}$ stores more detailed information on the distribution of  marked points on each boundary component.
One can determine the previous two kinds of spectra from  the boundary length and point spectrum by forgetting the partitions of marked points on the boundary cycles.

It is known that the enumeration of chord diagrams is intimately related to matrix models and cut-and-join equations \cite{ACPRS,ACPRS2,Andersen_new2,Dumitrescu:2012dka,MuSu}.
In this paper, the enumeration of partial chord diagrams labeled by the boundary length and point spectrum with the genus filtration is studied using matrix model techniques. 
Let ${\mathcal N}_{g,k,l}(\{b_i\},\{n_{\ii}\})$ denote the number of connected
 chord diagrams labeled by the set of parameters $(g,k,l;\{b_i\};\{n_{\ii}\})$.
We define the generating function of these numbers 
\begin{align}
\begin{split}
&
{\mathcal F}(x,y;\{s_i\};\{u_{\ii}\})=\sum_{b\ge 1}{\mathcal F}_b(x,y;\{s_i\};\{u_{\ii}\}),
\\
&
{\mathcal F}_b(x,y;\{s_i\};\{u_{\ii}\})=\frac{1}{b!}
\sum_{\sum_ib_i=b}\sum_{\{n_{\ii}\}}
{\mathcal N}_{g,k,l}(\{b_i\},\{n_{\ii}\})x^{2g-2}y^k
\prod_{i\ge 0}s_i^{b_i}\prod_{K\ge 1}\prod_{\{i_L\}_{L=1}^K} u_{\ii}^{n_{\ii}}.
\end{split}
\end{align}
Generating functions of disconnected and connected diagrams are related via the exponential relation
\begin{align}
&{\mathcal Z}(x,y;\{s_i\};\{u_{\ii}\})=\exp\left[{\mathcal F}(x,y;\{s_i\};\{u_{\ii}\})\right].
\end{align}
To analyze this enumeration further, we write the above generating function as a certain Hermitian matrix integral.
Let  ${\mathcal Z}_N(y;\{s_i\};\{u_{\ii}\})$ be the matrix integral over rank $N$ Hermitian matrices ${\mathcal H}_N$
\begin{align}
\begin{split}
&
{\mathcal Z}_N(y;\{s_i\};\{u_{\ii}\}) =
\\
&
=\frac{1}{\mathrm{Vol}_N}\int_{\mathcal{H}_N} dM\;\exp\bigg[
-N\mathrm{Tr}\bigg(\frac{M^2}{2}-\sum_{i\ge 0}s_i(y^{1/2}\Lambda_{\mathrm{L}}^{-1}M+\Lambda_{\mathrm{P}})^i\Lambda_{\mathrm{L}}^{-1}\bigg)
\bigg],
\end{split}
\end{align}
where $\Lambda_{\mathrm{P}}$ and $\Lambda_{\mathrm{L}}$ are {\em external  matrices} \cite{Kontsevich1} of rank $N$, and the normalization factor $\mathrm{Vol}_N$ is defined in (\ref{volN_def}).
In this matrix integral representation,
the counting parameter $u_{(i_1,\ldots,i_K)}$ is identified with the trace of the corresponding product of external matrices
\begin{align}
u_{(i_1,\ldots,i_K)}=\frac{1}{N}\mathrm{Tr}\left(
\Lambda_{\mathrm{P}}^{i_1}\Lambda_{\mathrm{L}}^{-1}\Lambda_{\mathrm{P}}^{i_2}\Lambda_{\mathrm{L}}^{-1}\cdots\Lambda_{\mathrm{P}}^{i_K}\Lambda_{\mathrm{L}}^{-1}
\right).
\end{align}
In Theorem \ref{Thm:blp_matrix}  we show that 
\begin{align}
{\mathcal Z}_N(y;\{s_i\};\{u_{\ii}\})=
{\mathcal Z}(N^{-1},y;\{s_i\};\{u_{\ii}\}).
\end{align}

\begin{figure}[h]
\begin{center}
   \includegraphics[width=120mm,clip]{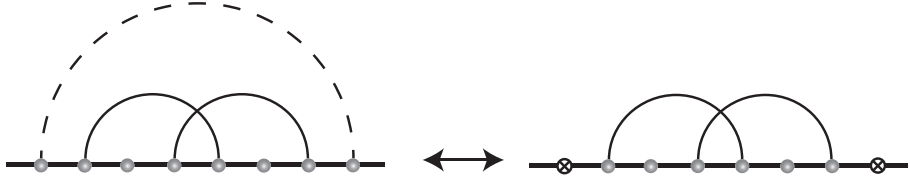}
\end{center}
\caption{\label{cut_join} The cut-and-join manipulations on chord diagrams.}
\end{figure}

This matrix integral representation provides a new, matrix model proof of the cut-and-join equation found by combinatorial means in \cite{Andersen_new}. The cut-and-join equation can be written as
\begin{align}
\frac{\partial}{\partial y} {\mathcal Z}(x,y;\{s_i\};\{u_{{\ii}}\})
={\mathcal M}{\mathcal Z}_N(x,y;\{s_i\};\{u_{{\ii}}\}),
\end{align}
where ${\mathcal M}$ is the second order partial differential operator in variables $u_{\ii}$ (see Theorem \ref{thm:cut_join_blp} for details). 
This cut-and-join equation can be regarded as the evolution equation in the variable $y$, and its formal solution reads 
\begin{align}
\begin{split}
&
{\mathcal Z}(x,y;\{s_i\};\{u_{{\ii}}\})=\mathrm{e}^{y{\mathcal M}}{\mathcal Z}(x,0;\{s_i\};\{u_{{\ii}}\}),
\\
&
{\mathcal Z}(x,0;\{s_i\};\{u_{{\ii}}\})=\mathrm{e}^{N^2\sum_{i\ge 0}s_iu_{(i)}}.
\end{split}
\end{align}

Expanding the operator $\mathrm{e}^{y{\mathcal M}}$ around $y=0$, 
one determines the number of connected partial chord diagrams ${\mathcal N}_{g,k,l}(\{b_i\},\{n_{\ii}\})$ iteratively from this formal solution.
The cut-and-join equation is a powerful method to  systematically count partial chord diagrams of a given length and point spectrum.

In this work we also generalize the above analysis to non-oriented analogues of partial chord diagrams.
By non-oriented partial chord diagrams we mean diagrams with all chords decorated by a binary variable, which indicates if they are {\em twisted} or not. When associating the surface ${\Sigma}_c$ to a non-oriented partial chord diagram, 
twisted bands are associated along the twisted chords as indicated in 
Figure \ref{partial_chord_non}.
This construction leads to $2^k$ orientable or non-orientable surfaces associated to one particular partial chord diagram with $k$ chords, if we consider all possible assignments of twisting or untwisting of $k$ bands. 
In the non-oriented case the Euler characteristic is defined as follows.
\begin{itemize}
\item Euler characteristic $\chi$.\\
The Euler characteristic of the two dimensional surface $\Sigma_c$
is defined by the formula
\begin{align}
\chi=2-h,
\nonumber
\end{align}
where $h$ is 
the number of cross-caps. The Euler relation holds
\begin{align}
 2-h=b-k+n.
\end{align}
\end{itemize}
With this setup, the enumeration of non-oriented partial chord diagrams can be considered analogously to the orientable case.

\begin{figure}[h]
\begin{center}
   \includegraphics[width=120mm,clip]{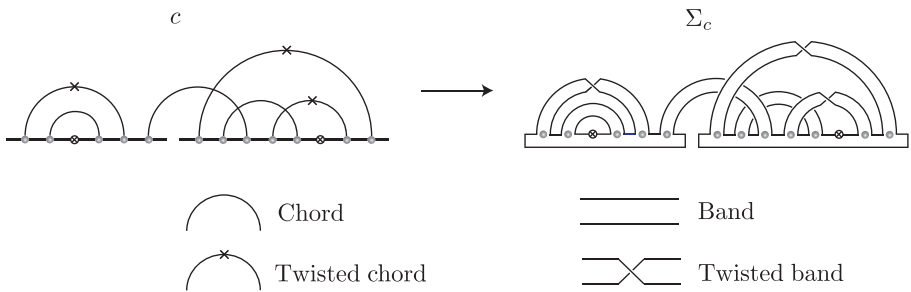}
\end{center}
\caption{\label{partial_chord_non} A non-oriented surface $\Sigma_c$ associated to a non-oriented partial chord diagram $c$.}
\end{figure}

Let $\widetilde{\mathcal N}_{h,k,l}(\{b_i\},\{n_{\ii}\})$ denote the number 
of connected non-oriented partial chord diagrams with the
cross-cap number $h$, $k$ chords, the backbone spectrum $\{b_i\}$, $l$ marked points,
 and  the boundary length and point spectrum $n_{\ii}$. The generating function
$\widetilde{\mathcal F}(x,y;\{s_i\};\{u_{\ii}\})$ is defined by
\begin{align}
\begin{split}
&
\widetilde{\mathcal F}(x,y;\{s_i\};\{u_{\ii}\})=\sum_{b\ge 1}
\widetilde{\mathcal F}_b(x,y;\{s_i\};\{u_{\ii}\}),\\
&
\widetilde{\mathcal F}_b(x,y;\{s_i\};\{u_{\ii}\})=\frac{1}{b!}
\sum_{\sum_ib_i=b}
\sum_{\{n_{\ii}\}}
\widetilde{\mathcal N}_{h,k,l}(\{b_i\},\{n_{\ii}\})x^{h-2}y^k
\prod_{i\ge 0}s_i^{b_i}\prod_{K\ge 1}\prod_{\{i_L\}_{L=1}^K} u_{\ii}^{n_{\ii}}.
\end{split}
\end{align}
We also define the generating function of the numbers 
of connected and disconnected 
non-oriented partial chord diagrams
\begin{align}
\widetilde{\mathcal Z}(x,y;\{s_i\};\{u_{\ii}\})=\exp\left[\widetilde{\mathcal F}(x,y;\{s_i\};\{u_{\ii}\})\right].
\end{align}
In Theorem \ref{thm:blp_matrix_non} we show that
this generating function can be expressed as a real symmetric
matrix integral with two external symmetric matrices $\Omega_{\mathrm{P}}$ and $\Omega_{\mathrm{L}}$ 
\begin{align}
& 
\widetilde{\mathcal Z}(N^{-1},y;\{s_i\};\{u_{\ii}\})=\widetilde{\mathcal Z}_N(y;\{s_i\};\{u_{\ii}\}),
\\
&
\widetilde{\mathcal Z}_N(y;\{s_i\};\{u_{\ii}\}) =
\nonumber \\
&
=\frac{1}{\mathrm{Vol}_N(\mathbb{R})}\int_{{\mathcal H}_N(\mathbb{R})} dM\;\exp\bigg[
-N\mathrm{Tr}\bigg(\frac{M^2}{4}-\sum_{i\ge 0}s_i(y^{1/2}\Omega_{\mathrm{L}}^{-1}M+\Omega_{\mathrm{P}})^i
\Omega_{\mathrm{L}}^{-1}\bigg)
\bigg],
\end{align}
where the normalization factor $\mathrm{Vol}_N(\mathbb{R})$ is defined in (\ref{volN_def_r}), and ${\mathcal H}_N(\mathbb{R})$ is the space of real symmetric matrices of rank $N$. The parameter $u_{(i_1,\ldots,i_K)}$ is 
identified with a trace of the external matrices via the formula
\begin{align}
u_{(i_1,\ldots,i_K)}=\frac{1}{N}\mathrm{Tr}\left(
\Omega_{\mathrm{P}}^{i_1}\Omega_{\mathrm{L}}^{-1}\Omega_{\mathrm{P}}^{i_2}\Omega_{\mathrm{L}}^{-1}\cdots\Omega_{\mathrm{P}}^{i_K}\Omega_{\mathrm{L}}^{-1}
\right).
\end{align}

Using this matrix integral representation of the generating function, 
one can again prove the cut-and-join equation, established independently by combinatorial arguments in \cite{Andersen_new}
\begin{align}
&\frac{\partial}{\partial y} \widetilde{\mathcal Z}_N(y;\{s_i\};\{u_{{\ii}}\})=
\widetilde{\mathcal M}\widetilde{\mathcal Z}_N(y;\{s_i\};\{u_{{\ii}}\}),
\end{align}
where $\widetilde{\mathcal M}$ is a second order partial differential operator in the variables $u_{\ii}$.
The details of the differential operator $\widetilde{\mathcal M}$ and the matrix model derivation 
of the cut-and-join equation are presented in Theorem \ref{Thm:cut_join_non}.

\subsection{Motivation: RNA chains}

One important motivation to study partial chord diagrams in this and the preceding work \cite{AAPZ,Andersen_new} is a complicated problem of RNA structure prediction in molecular biology, which we now shortly review.

An RNA molecule is a linear polymer, referred to as the backbone, that consists of four types of nucleotides: adenine, cytosine, guanine, and uracil, denoted respectively ${\bf A}$, ${\bf C}$, ${\bf G}$, and ${\bf U}$. The backbone is endowed with an orientation from 5'-end to 3'-end, and the primary sequence is the sequence of nucleotides read with respect to this orientation.
Between nucleotides hydrogen bonds are formed, resulting in the so-called Watson-Click pairs involving ${\bf A-U}$ or ${\bf G-C}$ nucleotides; in addition Wobble pairs ${\bf U-G}$ can be formed. The set of base pairs formed by such hydrogen bonds is referred to as the secondary structure.\footnote{There are other types of interactions in RNA secondary structure, which are however less common and we ignore them in this discussion.} Prediction of the secondary structure from the primary sequence is an outstanding problem that was initiated by the pioneering work of Michael Waterman \cite{Waterman0} and has been studied intensively for last three decades.

\begin{figure}[h]
\begin{center}
   \includegraphics[width=80mm,clip]{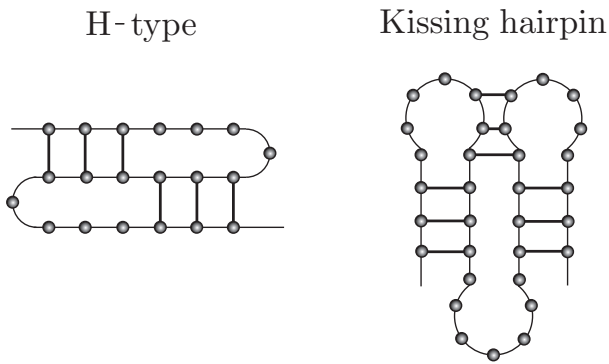}
\end{center}
\caption{\label{pseudo} Pseudoknot structures in RNA. The long curved line, blobs (i.e. marked points), and short lines represent the backbone, nucleotides, and base pairs, respectively.}
\end{figure}

Topologically, we can represent the base pairings for a given RNA structure by a partial chord diagram as follows. The backbone is represented as a disjoint union of horizontal straight line segments (arranged along the real line in the plane), one for each backbone component, and 
each nucleotide is represented as a marked point on this union of line segments. 
The base pairs are represented by chords in the upper-half plane attached at
two marked points corresponding to the bonded pair of nucleotides.

Note that a partial chord diagram has genus zero if no two of its chords cross each other. If however such crossings exist, then the structure is referred to as a pseudoknot, and its genus is non-zero.
Considerable number of pseudoknot structures have been observed,
e.g. tRNAs, RNAseP \cite{LoPa}, telomerase RNA \cite{SB} and ribosomal RNAs \cite{KG}. According to the online database ``RNA-strand'' half of the known structures form pseudoknots \cite{RNA-strand}. There are various kinds of pseudoknots classified by the topology of the RNA \cite{APRWat}, referred to as e.g. H-type \cite{Akutsu}, kissing hairpin \cite{CCJ,RE}, etc. 

In recent years, a combinatorial description of RNA structures in terms of linear chord diagrams has been developed in a series of works \cite{OZ,VOZ,VO, APRWat,APRWan,AHPR,ACPRS,ACPRS2,AAPZ,Reidys,P5}.
However, a large class of reasonable energy-based models that predict the secondary structure including pseudoknots are NP complete \cite{LyPe, Akutsu}, 
and a fully satisfactory energy model for RNA, including pseudoknot structures, has not been established yet.

\begin{figure}[h]
\begin{center}
   \includegraphics[width=90mm,clip]{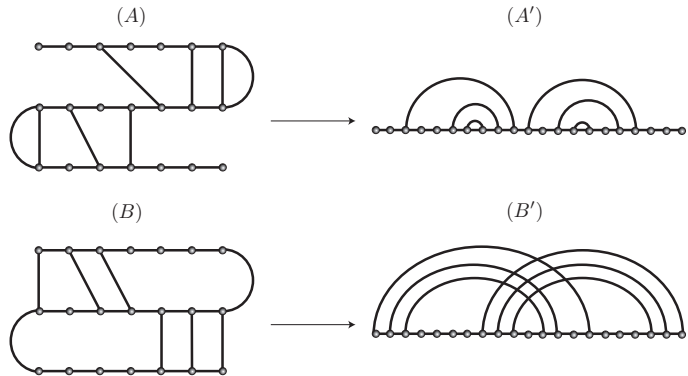}
\end{center}
\caption{\label{Cross_RNA} Partial chord diagrams unveil the difference in the topological structure of RNA molecules.}
\end{figure}

In the search of a realistic energy function for RNA structures with pseudoknots,
the boundary length and point spectrum should provide a useful tool that includes more detailed information about the location of marked points.
In standard algorithms developed by Waterman \cite{Waterman}, Nussinov et al. \cite{Nussinov}, Zucker and Stiegler \cite{Zucker}, etc., dynamic programming (DP) 
has been used to predict most likely secondary structures. 
Indeed, in famous algorithms such as \cite{Zucker0,Vienna}, the (loop-based) energy in each configuration of RNA is considered. In these algorithms, the most probable secondary structure is determined as the minimum free energy configuration, and to make them more efficient the statistical mechanical ensemble (i.e. the partition function algorithm) is implemented \cite{McCaskill}. The application of these algorithms, which include pseudoknot structures stratified by $\gamma$ {\em structures},
was studied in \cite{RHAPSN,Reidys}.
Most of the energy functions essentially respect the boundary point and length spectra independently.
In order to improve the energy model for RNA structure prediction with pseudoknots,
it would be useful to explore energy parameters for more realistic and efficient energy function 
on the basis of the boundary length and point spectrum.

\subsection{Plan of the paper}

This paper is organized as follows. In Section \ref{sec:const_mat} we
construct Hermitian matrix models with external matrices,
which encode generating functions of orientable partial chord diagrams labeled by the boundary point spectrum 
(in Subsection \ref{subsec:bp}), the boundary length spectrum (in Subsection \ref{subsec:bl}), and 
the boundary length and point spectrum (in Subsection \ref{subsec:bpl}). All these 
constructions are established by the correspondence between chord diagrams and Wick contractions via the Wick theorem. The matrix model encoding the boundary 
length and point spectrum is given in Theorem \ref{Thm:blp_matrix}.
In Section \ref{sec:mat_heat} we derive partial differential equations for 
matrix integrals found in Section \ref{sec:const_mat}. These partial 
differential equations coincide with the cut-and-join equations found 
combinatorially in \cite{AAPZ, Andersen_new}. The cut-and-join equation for 
partial chord diagrams labeled by the boundary length and point spectrum is 
determined in Theorem \ref{thm:cut_join_blp}.
Section \ref{sec:non_orient} is devoted to the analysis of non-oriented analogues of the 
results obtained in Section \ref{sec:const_mat} and \ref{sec:mat_heat}. 
In Subsection \ref{subsec:non_or_mat} we find 
real symmetric matrix models with external matrices, that encode generating functions of both orientable and non-orientable partial chord diagrams.
The non-oriented analogue of the matrix integral from Theorem \ref{Thm:blp_matrix} is given in Theorem \ref{thm:blp_matrix_non}. 
Non-oriented analogues of cut-and-join equations from Section 
\ref{sec:mat_heat} are determined in Theorem \ref{Thm:cut_join_non}.
In Appendix \ref{app:heat_no} we derive a partial differential equation
from Proposition \ref{prop:master_heat_r} for a real symmetric matrix integral 
with external matrices.
In Appendix \ref{app:miwa_no} we prove Lemma \ref{lem:ext_Miwa_non}.

\section{\label{sec:const_mat}Enumerating partial linear chord diagrams via matrix models}

The enumeration problem of partial chord diagrams with respect to the genus filtration has been reformulated in terms of matrix integrals. Matrix model techniques for enumeration of the RNA structures with pseudoknots have been developed in a series of papers \cite{OZ,VOZ,VO}, and independently in \cite{ACPRS,ACPRS2,Andersen_new2}. Subsequently the analysis involving boundary point and length spectra of partial linear chord diagrams has been conducted in \cite{AAPZ,Andersen_new}.
In this section we develop a new perspective on this problem and construct a matrix model that enumerates partial chord diagrams labeled by the boundary length and point spectrum.

\subsection{\label{subsec:bp}A matrix model enumerating partial chord diagrams}

In the first step we construct a matrix model that counts partial chord diagrams labeled by the boundary point spectrum $\{n_i\}$. 
\begin{definition}\label{definition_notations}
Let ${\mathcal N}_{g,k,l}(\{b_i\},\{n_i\},\{p_i\})$ denote the number of
connected partial chord diagrams of type $\{g, k, l; \{b_i\};\{n_i\};\{p_i\}\}$.
In particular, focusing on the boundary point
spectrum we define the
following number of partial chord diagrams characterized by the data
$\{g,k,l;\{b_i\},\{n_i\}\}$,
\begin{align}
{\mathcal N}_{g,k,l}(\{b_i\},\{n_i\})=\sum_{\{p_i\}}{\mathcal N}_{g,k,l}(\{b_i\},\{n_i\},\{p_i\}).
\nonumber
\end{align}
We introduce the generating function\footnote{
The parameters $s_i$ and $t_i$ in this article and in \cite{AAPZ} are related by
$s_i \leftrightarrow t_i$.}  for the numbers ${\mathcal N}_{g,k,l}(\{b_i\},\{n_i\})$
\begin{align}
\begin{split}
&
F(x,y;\{s_i\};\{t_i\})=\sum_{b\ge 1} F_b(x,y;\{s_i\};\{t_i\}),
\\
&
F_b(x,y;\{s_i\};\{t_i\})=\frac{1}{b!}\sum_{\sum_ib_i=b}\sum_{\{n_i\}}{\mathcal N}_{g,k,l}(\{b_i\},\{n_i\})x^{2g-2}y^k
\prod_{i\ge 0}s_i^{b_i}t_i^{n_i}.
\label{free_energy_cut_join2}
\end{split}
\end{align}
\end{definition}
The generating function for the numbers $\widehat{\mathcal N}_{k,b,l}(\{b_i\},\{n_i\})$ 
of connected and disconnected partial chord diagrams 
arises in the usual way from the exponent
\begin{align}
Z^{\mathrm{P}}(x,y;\{s_i\};\{t_i\})&=\exp\left[
F(x,y;\{s_i\};\{t_i\})\right]
\nonumber \\
&
=\sum_{\{b_i\}}\sum_{\{n_i\}}
\widehat{\mathcal N}_{k,b,l}(\{b_i\},\{n_i\})x^{-b+k-n}y^k
\prod_{i\ge 0}s_i^{b_i}t_i^{n_i}.
\label{partition_function_cut_join}
\end{align}
In the following we rewrite the generating function $Z^{\mathrm{P}}(x,y;\{s_i\};\{t_i\})$ as a Hermitian matrix integral. To this end, we consider first Gaussian averages over Hermitian matrices.

\begin{definition}
Let $\mathcal{O}(M)$ be a function of a rank $N$ Hermitian matrix $M$.
The Gaussian average $\langle {\mathcal O}(M)\rangle_N^{\mathrm{G}}$ is defined by the integral over the space ${\mathcal H}_N$ of rank $N$ Hermitian matrices with respect
to the Haar measure $dM$ with the Gaussian weight $\mathrm{e}^{-N\mathrm{Tr}\frac{M^2}{2}}$,
\begin{align}
\langle {\mathcal O}(M)\rangle_N^{\mathrm{G}}=\frac{1}{\mathrm{Vol}_N}
\int_{{\mathcal H}_N}dM\;{\mathcal O}(M)\,\mathrm{e}^{-N\mathrm{Tr}\frac{M^2}{2}},
\label{Gauss_average}
\end{align}
where the normalization factor $\mathrm{Vol}_N$ takes form
\begin{align}
\mathrm{Vol}_N=\int_{{\mathcal H}_N} dM\;\mathrm{e}^{-N\mathrm{Tr}\frac{M^2}{2}}
=N^{N(N+1)/2}\mathrm{Vol}({\mathcal H}_N).
\label{volN_def}
\end{align}
In particular for  $\mathcal{O}(M)=M_{\alpha\beta}M_{\gamma\epsilon}$ ($\alpha,\beta,\gamma,\epsilon=1,\ldots,N$),
the Gaussian average is 
\begin{align}
\contraction[1ex]{}{xX}{X}{}
M_{\alpha\beta}M_{\gamma\delta}
:=\langle M_{\alpha\beta}M_{\gamma\epsilon}\rangle_N^{\mathrm{G}}
=\frac{1}{N}\delta_{\alpha\epsilon}\delta_{\beta\gamma}.
\label{Wick_contr}
\end{align}
This quantity is called the {\it Wick contraction}. By definition, a multiple Wick contraction is a product of the Gaussian average of each Wick contracted pair.
\end{definition}

It follows from the definition (\ref{Gauss_average}) that Gaussian averages of an odd number of matrix elements vanish. On the other hand, Gaussian averages of an even number of matrix elements are non-zero, and can be computed using the \textit{Wick theorem} \cite{BIZ,P2,Mu}, as we now recall. Consider an ordered sequence 
\begin{align}
M_{\alpha_1\beta_1}M_{\alpha_2\beta_2}\cdots M_{\alpha_{2k}\beta_{2k}}
\nonumber
\end{align}
of $2k$ matrix elements $M_{\alpha_n\beta_n}$ ($n=1,\ldots,2k$).

Let $P_k$ denote a set of matchings by $k$ Wick contractions among the $2k$ matrix elements in the above sequence. $P_k$ is isomorphic to the following quotient of groups
\begin{align}
P_k\simeq G_H/G_E,\quad G_H=S_{2k},\quad G_E=S_k\rtimes (S_2)^k.
\nonumber
\end{align}
Here the elements of the permutation group $S_{2k}$ permute $2k$ matrix elements.
The factors $S_k$ of $G_E$ act by permuting $k$ Wick contractions and $(S_2)^k$ swaps matrix elements in each Wick contracted pair.
The Wick theorem implies the following result.
\begin{theorem}\label{thm:Wick}
The Gaussian average of $2k$ matrix elements $M_{\alpha_n\beta_n}$ ($n=1,\ldots,k$) 
equals 
\begin{align}
\begin{split}
\langle
M_{\alpha_1\beta_1}M_{\alpha_2\beta_2}\cdots M_{\alpha_{2k}\beta_{2k}}
\rangle^{\mathrm{G}}_N
&=
\sum_{\sigma\in P_k}\prod_{i=1}^k
\contraction[1ex]{}{X}{XXXXXXX}{}
M_{\alpha_{\sigma(2i-1)}\beta_{\sigma(2i-1)}}M_{\alpha_{\sigma(2i)}\beta_{\sigma(2i)}}
\\
&=\frac{1}{N^k}\sum_{\sigma\in P_k}\prod_{i=1}^k\delta_{\alpha_{\sigma(2i-1)}\beta_{\sigma(2i)}}
\delta_{\alpha_{\sigma(2i)}\beta_{\sigma(2i-1)}}.
\label{Wick}
\end{split}
\end{align}
\end{theorem}

\subsubsection{Chord diagrams and Wick contractions} 

Let $c$ be a chord diagram.
We now recall the explicit relation between a surface $\Sigma_c$ associated to a chord diagram $c$ and $k$-matchings or Wick contractions in the Gaussian average. To illustrate this correspondence we depict chord ends on backbones in $\Sigma_c$
as trivalent vertices that consist of upright and horizontal line segments, see Figure \ref{rule_C}.
This correspondence is specified by the following four points \textbf{C1}--\textbf{C4}.

\begin{figure}[h]
\begin{center}
   \includegraphics[width=100mm,clip]{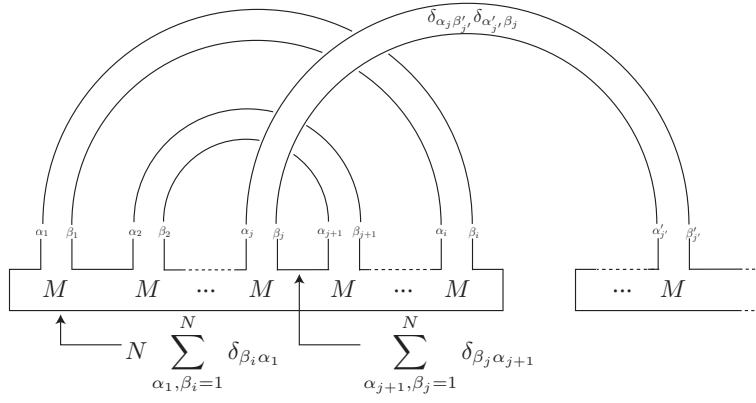}
\end{center}
\caption{\label{rule_C} Bijective correspondence between chord diagrams and Wick contractions.}
\end{figure}
\begin{description}
\item[C1] A matrix element $M_{\alpha\beta}$ corresponds to a chord end on a backbone.
Indices $\alpha,\beta(=1,\ldots,N)$ are assigned to two upright line segments on the upper edge of the backbone.

\item[C2] If two matrix elements $M_{\alpha_j\beta_j}M_{\alpha_{j+1}\beta_{j+1}}$ correspond to two adjacent chord ends on the same backbone, 
then the following quantity is assigned to the horizontal segment between these two chord ends on the upper edge of the backbone
\begin{align}
\sum_{\alpha_{j+1},\beta_j=1}^N\delta_{\beta_j\alpha_{j+1}}.
\nonumber
\end{align}
This assignment encodes matrix multiplication of matrix elements corresponding to adjacent chord ends on the backbone.

\item[C3] For the product of $i$ matrix elements $M$
\begin{align}
\sum_{\alpha_2,\ldots,\alpha_{i}=1}^N
\sum_{\beta_1,\ldots,b_{i-1}=1}^NM_{\alpha_1\beta_1}\delta_{\beta_1\alpha_2}M_{\alpha_2\beta_2}\ldots \delta_{\beta_{i-1}\alpha_i}M_{\alpha_{i}\beta_{i}}
=(M^i)_{\alpha_1\beta_i},
\nonumber
\end{align}
which corresponds to a backbone with $i$ chord ends,
the following quantity is assigned to the bottom edge of the backbone
\begin{align}
N\sum_{\alpha_1,\beta_i=1}^N\delta_{\beta_i\alpha_1}.
\nonumber
\end{align}
Thus, a backbone with $i$ chord ends corresponds to a single trace of the $i$'th power of $M$, namely $N\mathrm{Tr}M^i$.

\item[C4] The Wick contraction between $M_{\alpha_j\beta_j}$ and $M_{\alpha'_{j'}\beta'_{j'}}$ corresponds to a band connecting two chord ends. 
Each Wick contraction imposes a constraint $\delta_{\alpha_{j}\beta'_{j'}}\delta_{\alpha'_{j'}\beta_j}$
on matrix indices assigned to the edges of the chord ends 
matched by the Wick contraction.
\end{description}

The above rules imply the following bijective correspondence 
\begin{align}
W^{\rm C}_N(\{b_i\})=\Bigl\langle \prod_{i}\left(N\mathrm{Tr}M^i\right)^{b_i}\Bigr\rangle_N^{\mathrm{G}},
\qquad \sum_iib_i=2k,
\label{complete_Gauss2}
\end{align}
between matchings by $k$ Wick contractions in the Gaussian average on one hand, and chord diagrams that consist of $b_i$ backbones with $i$ chord ends on the other hand, see Figure \ref{chord_wick}.
\begin{figure}[h]
\begin{center}
   \includegraphics[width=100mm,clip]{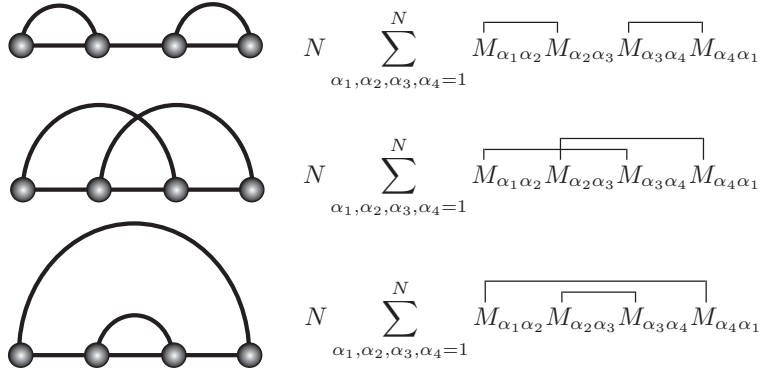}
\end{center}
\caption{\label{chord_wick}  Chord diagrams and Wick contractions for $\langle N\mathrm{Tr} M^4\rangle_{\mathrm{N}}^{\mathrm{G}}$.}
\end{figure}

The Wick contractions (\ref{Wick}) in $W^{\rm C}_N(\{b_i\})$ replace
all matrix elements $M$'s by products of $\delta$'s,
and summing over matrix indices along a boundary cycle
one finds a factor of $N$ corresponding to each boundary cycle in a chord diagram.
Therefore the overall $N$ dependence following from the above rules
amounts to assigning $N^{b-k+n}$ factor to the term $W^{\rm C}_N(\{b_i\})$, 
corresponding to a chord diagram with backbone spectrum $\{b_i\}$ and $n$ boundary cycles.
Combing the Wick theorem and this bijective correspondence between 
matchings by $k$ Wick contractions in the Gaussian average $W^{\rm C}_N(\{b_i\})$ and the set of chord
diagrams with backbone spectrum $\{b_i\}$, the following proposition follows.
\begin{proposition}\label{prop:complete}
The Gaussian average $W^{\rm C}_N(\{b_i\})$ in equation (\ref{complete_Gauss2}) agrees with the generating function of chord diagrams with backbone spectrum $\{b_i\}$
\begin{align}
W^{\rm C}_N(\{b_i\})=
\sum_{n\ge 0}\widehat{\mathcal N}_{k,b,n}(\{b_i\})N^{b-k+n}.
\end{align}
\end{proposition}

Here $\widehat{\mathcal N}_{k,b,n}(\{b_i\})$ is the number of chord diagrams that consist of $b_i$ backbones with $i$ trivalent vertices 
\begin{align}
\widehat{\mathcal N}_{k,b,n}(\{b_i\})=\sum_{\{p_i\}}\widehat{\mathcal N}_{k,b,l=0}(\{b_i\},n_0=n,\{n_i=0\}_{i\ge 1},\{p_i\}).
\end{align}

\subsubsection{Partial chord diagrams and Wick contractions}

We now generalize the above bijective correspondence to partial chord diagrams.
Let $c$ be a partial chord diagram. On the boundary cycles of the surface $\Sigma_c$ we add additional marked points, which correspond to those marked points on $c$ which are not chord ends.
These marked points are represented by \textit{external matrices} $\Lambda_{\mathrm{P}}$ of rank $N$ in the Gaussian average.
The rules \textbf{P1}--\textbf{P5} below provide the correspondence between 
partial chord diagrams with backbone spectrum $\{b_i\}$
and matchings with $k$ Wick contractions in the Gaussian average.
\begin{description}
\item[P1] A matrix element $M_{\alpha\beta}$ corresponds to a chord end on a backbone.
The graphical rule is the same as the rule $\mathbf{C1}$.
 
\item[P2] A matrix element $\Lambda_{\mathrm{P}\alpha\beta}$ corresponds to a marked point on a backbone in $\Sigma_c$.
Indices $\alpha,\beta(=1,\ldots,N)$ are assigned to two upright line segments at each marked point on the upper edge of the backbone, see Figure \ref{rule_P}.

\item[P3] To a line segment (on the upper edge of the backbone) between adjacent chord ends or marked points (located on the same backbone), corresponding to matrix elements $U_{\alpha_{j}\beta_{j}}$ and $V_{\alpha_{j+1}\beta_{j+1}}$ (for $U,V=M$ or $\Lambda_{\mathrm{P}}$), we assign 
\begin{align}
\sum_{\beta_j,\alpha_{j+1}=1}^N\delta_{\beta_j\alpha_{j+1}},
\end{align}
just as in {\bf C2}.  

\item[P4] Let $v_{j},w_{j}\in \mathbb{Z}_{\ge 0}$ ($j=1,\ldots,i$) with $\sum_{j=1}^i (v_j+w_j)=i$.
For an ordered matrix product  
\begin{align}
(M^{v_1}\Lambda_{\mathrm{P}}^{w_1}M^{v_2}\Lambda_{\mathrm{P}}^{w_2}
\cdots M^{v_i}\Lambda_{\mathrm{P}}^{w_i}
)_{\alpha_1\beta_i},
\label{order_P}
\end{align}
corresponding to a backbone which is an ordered sequence of $v_j$ chord ends and $w_j$ marked points,
we assign 
\begin{align}
N\sum_{\alpha_1,\beta_i=1}^N\delta_{\beta_i\alpha_1}
\nonumber
\end{align}
to the bottom edge of this backbone. It follows that the trace
\begin{align}
N\mathrm{Tr}(M^{v_1}\Lambda_{\mathrm{P}}^{w_1}M^{v_2}\Lambda_{\mathrm{P}}^{w_2}
\cdots M^{v_i}\Lambda_{\mathrm{P}}^{w_i})
\label{backbone_P}
\end{align}
is assigned to this backbone.
\item[P5] The Wick contraction between $M_{\alpha_j\beta_j}$ and $M_{\alpha'_{j'}\beta'_{j'}}$ corresponds to a band connecting two chord ends, and it is represented in the same way as  specified in \textbf{C4}.
\end{description}

\begin{figure}[h]
\begin{center}
   \includegraphics[width=120mm,clip]{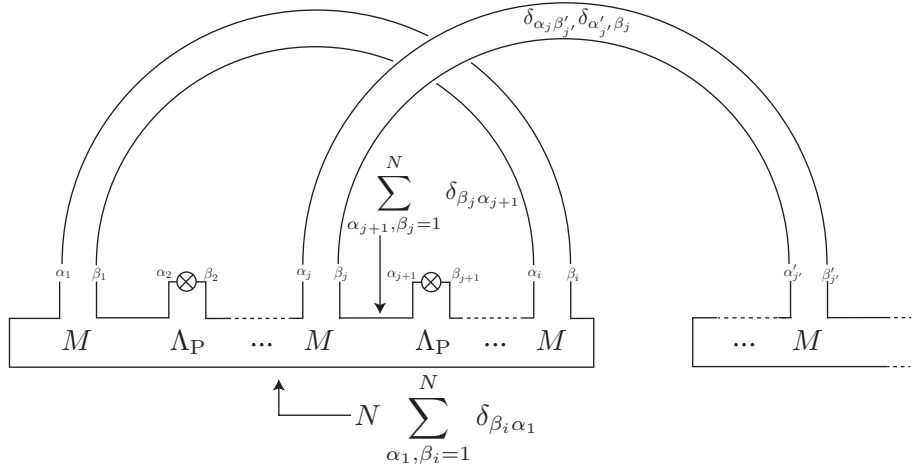}
\end{center}
\caption{\label{rule_P} Bijective correspondence between partial chord diagrams and matchings of Wick contractions}
\end{figure}

For a fixed backbone spectrum $\{b_i\}$, all possible sequences $\{\alpha_j,\beta_j\}$
in the expression (\ref{backbone_P}) are generated by the following product of traces
\begin{align}
\prod_{i\ge 0}\left(
N\mathrm{Tr}(M+\Lambda_{\mathrm{P}})^i\right)^{b_i}.
\end{align}
Hence, by the above rules, all partial chord diagrams with the backbone spectrum $\{b_i\}$ correspond bijectively to all matchings
by Wick contractions among the $M$'s in the expansion of the Gaussian average 
\begin{align}
W_N^{\mathrm{P}}(\{b_i\},\{r_i\})
=\Bigl\langle
\prod_{i\ge 0}\left(
N\mathrm{Tr}(M+\Lambda_{\mathrm{P}})^i\right)^{b_i}
\Bigr\rangle_N^{\rm G},
\label{I_N}
\end{align}
where we introduced the reverse Miwa times
\begin{align}
r_i=\frac{1}{N}\mathrm{Tr}\Lambda_{\mathrm{P}}^i.
\label{Miwa1}
\end{align}

If there are $n_i$ boundary components containing $i$ marked points, then one finds
a trace factor $(\mathrm{Tr}\Lambda_{\mathrm{P}}^i)^{n_i}$ in the corresponding term in the Gaussian average (\ref{I_N}), see Figure \ref{m4k1}.
Therefore, for partial chord diagrams with the backbone spectrum $\{b_i\}$ and the boundary point spectrum $\{n_i\}$,
the corresponding term in $W_N^{\mathrm{P}}(\{b_i\},\{r_i\})$ 
contributes the factor
\begin{align}
N^{b-k+n}\prod_{i\ge 0}r_i^{n_i}.
\nonumber
\end{align}

\begin{figure}[h]
\begin{center}
   \includegraphics[width=80mm,clip]{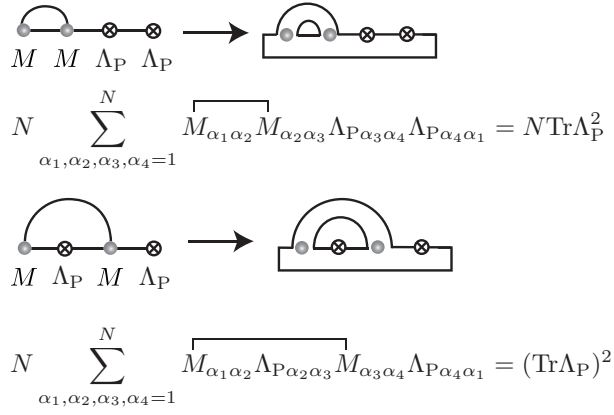}
\end{center}
\caption{\label{m4k1} Partial chord diagrams of types $\{g=0,k=1,l=2;b_4=1;n_0=1,n_2=1\}$ and $\{g=0,k=1,l=2;b_4=1;n_1=2\}$, and the corresponding Wick contractions.}
\end{figure}

Therefore, from Wick theorem and the above bijective correspondence between partial chord diagrams and matchings by Wick contractions, one finds the following proposition.
\begin{proposition}\label{prop:iI_N}
The Gaussian average  (\ref{I_N}) is the generating function 
for the numbers $\widehat{\mathcal N}_{k,b,l}({\{b_i\},\{n_i\}})$ of partial chord diagrams 
with the backbone spectrum $\{b_i\}$ and the boundary point spectrum $\{n_i\}$
\begin{eqnarray}
W_N^{\mathrm{P}}(\{b_i\},\{r_i\})=\sum_{\{n_i\}}
\widehat{\mathcal N}_{k,b,l}({\{b_i\},\{n_i\}})N^{b-k+n}\prod_{i\ge 0}r_i^{n_i},
\label{point_WN}
\end{eqnarray}
where the summation is constrained by $\sum in_i=\sum ib_i-2k$.
\end{proposition}

Using this proposition, we consider the full generating function $Z_N^{\mathrm{P}}(y;\{s_i\};\{r_i\})$ for the numbers $\widehat{\mathcal N}_{k,b,l}({\{b_i\},\{n_i\}})$ of partial chord diagrams weighted by 
\begin{align}
N^{b-k+n}y^k\prod_{i\ge 0}s_i^{b_i}r_i^{n_i}.
\nonumber
\end{align}
Since the contribution from a partial chord diagram is invariant under permutations of its backbones, the full generating function
\begin{align}
Z_N^{\mathrm{P}}(y;\{s_i\};\{r_i\})
=\sum_{\{b_i\}}\sum_{\{n_i\}}
\widehat{\mathcal N}_{k,b,l}(\{b_i\},\{n_i\})N^{b-k+n}y^k
\prod_{i\ge 0}s_i^{b_i}
r_i^{n_i}
\nonumber
\end{align}
can be rewritten as a sum over all backbone spectra $\{b_i\}$ of the terms 
$$y^{\sum_iib_i/2}W_N^{\mathrm{P}}(\{b_i\},\{y^{-i/2}r_i\})\prod_i\frac{s_i^{b_i}}{b_i!}.$$ 
It follows that
\begin{align}
Z_N^{\mathrm{P}}(y;\{s_i\};\{r_i\})
=\sum_{\{b_i\}}\prod_{i\ge 0}\frac{s_i^{b_i}y^{ib_i/2}}{b_i!}
\Big\langle
\Big(N\mathrm{Tr}(M+y^{-1/2}\Lambda_{\mathrm{P}})^i\Big)^{b_i}
\Big\rangle_N^{\mathrm{G}}.
\nonumber
\end{align}
Performing the summation over $b_i$'s,  one finds that the full generating function is given by the matrix integral
\begin{align}
\begin{split}
&
Z_N^{\mathrm{P}}(y;\{s_i\};\{r_i\}) =
\\
&
=\frac{1}{\mathrm{Vol}_N}\int_{{\mathcal H}_N} dM\;\exp\bigg[
-N\mathrm{Tr}\bigg(\frac{M^2}{2}-
\sum_{i\ge 0}s_i(y^{1/2}M+\Lambda_{\mathrm{P}})^i\bigg)
\bigg].
\label{partition_fn_pt}
\end{split}
\end{align}

This matrix integral and $Z^{\mathrm{P}}(x,y;\{s_i\};\{t_i\})$ in equation (\ref{partition_function_cut_join}) are identified  by a change of variables.
Since the reverse Miwa time for $i=0$ yields $r_0=1$ automatically, 
we need to introduce the parameter $t_0$ by the following change of variables
\begin{align}
N\to t_0 N,\quad y\to t_0y,\quad s_i\to t_0^{-1}s_i,\quad 
r_i\to t_0^{-1}t_i.
\nonumber
\end{align}
As a result, we find the main theorem in this subsection.
\begin{theorem}\label{thm:point_ori}
The generating function (\ref{partition_function_cut_join}) 
is given by the matrix integral (\ref{partition_fn_pt}),
\begin{equation}
Z^{\mathrm{P}}(N^{-1},y;\{s_i\};\{t_i\})=Z^{\mathrm{P}}_{t_0 N}(t_0 y;\{t_0^{-1}s_i\};\{t_0^{-1}t_i\}).
\label{rescale_pt}
\end{equation}
\end{theorem}

\subsection{\label{subsec:bl}A matrix model for the enumeration of chord diagrams}

Next we turn to the enumeration of chord diagrams
labeled by the backbone spectrum $\{b_i\}$ and the boundary length spectrum $\{p_i\}$. 
The number ${\mathcal N}_{g,k}(\{b_i\},\{p_i\})$ of connected chord diagrams is given by
\begin{align}
{\mathcal N}_{g,k}(\{b_i\},\{p_i\})
=\sum_{\{n_i\}}{\mathcal N}_{g,k,0}(\{b_i\},\{n_i\},\{p_i\}).
\nonumber
\end{align}
We introduce the following generating function of these numbers\footnote{
The parameters $q_i$'s in our paper correspond to $s_i$'s in \cite{AAPZ}.
}
\begin{definition}\label{def:generating_length_ori}
Let $G(x,y;\{s_i\};\{q_i\})$ denote the generating function of chord diagrams labeled by the boundary length spectrum
\begin{align}
\begin{split}
&
G(x,y;\{s_i\};\{q_i\})=\sum_{b\ge 1}G_b(x,y;\{s_i\};\{q_i\}),
\\
&
G_b(x,y;\{s_i\};\{q_i\})=\frac{1}{b!}\sum_{\sum b_i=b}\sum_{\{p_i\}}
{\mathcal N}_{g,k}(\{b_i\},\{p_i\})x^{2g-2}y^k\prod_{i\ge 0}s_i^{b_i}
\prod_{i\ge 1}q_i^{p_i}.
\label{free_length_ori}
\end{split}
\end{align}
In the same way as the generating function $Z^{\mathrm{P}}(x,y;\{s_i\};\{t_i\})$ in (\ref{partition_function_cut_join}), the generating function for the numbers 
$\widehat{\mathcal N}_{k,b}(\{b_i\},\{p_i\})$
of connected and disconnected chord takes form
\begin{align}
Z^{\mathrm{L}}(x,y;\{s_i\};\{q_i\})
&=\mathrm{exp}\left[G(x,y;\{s_i\};\{q_i\})\right]
\nonumber \\
&
=\sum_{\{b_i\}}\sum_{\{p_i\}}
\widehat{\mathcal N}_{k,b}(\{b_i\};\{p_i\})x^{-b+k-n}y^k
\prod_{i\ge 0}s_i^{b_i}\prod_{i\ge 1}q_i^{p_i}.
\label{partition_length_ori}
\end{align}
\end{definition}

\subsubsection{A matrix model for the boundary length spectrum}

Let $c$ be a chord diagram.
The boundary length spectrum filters chord diagrams
according to combinatorial length of each boundary cycle, i.e. the sum of the number of chords and backbone underpasses.
This length can be determined by counting marked points of a new type, which we now introduce.
We introduce marked points of a new type between all chord ends and backbone ends, see the left diagram in Figure \ref{fat_length_ext0}.
For chord diagram decorated in this way, we get new marked points on the boundaries of the surface $\Sigma_c$ by sliding each new marked point along the boundary of $\Sigma_c$ until it reaches the first chord or backbone underside midpoint, as indicated in the right hand side of Figure \ref{fat_length_ext0}.
\begin{figure}[h]
\begin{center}
   \includegraphics[width=120mm,clip]{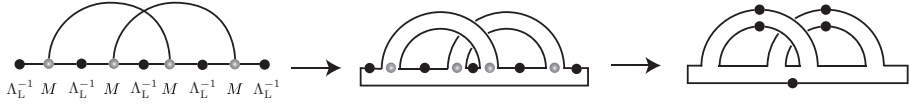}
\end{center}
\caption{\label{fat_length_ext0} Decorating a chord diagram with new marked points for partitions.}
\end{figure}

In order to construct a Gaussian matrix integral which counts this type of chord diagrams we introduce another external matrix $\Lambda_{\mathrm{L}}$, which is an invertible rank $N$ matrix that keeps track of new marked points. 
We introduce a new model model based on the following rules \textbf{L1}--\textbf{L5}, in which Wick contractions in the Gaussian average correspond bijectively to decorated chord diagrams.
\begin{description}
\item[L1] A matrix element $M_{\alpha\beta}$ corresponds to a chord end on a backbone.
This graphical rule is the same as the rule \textbf{C1}.
\item[L2]  
A matrix element $(\Lambda_{\mathrm{L}}^{-1})_{\alpha_j\beta_j}$ is adjacent to a matrix element
$M_{\alpha_{j+1}\beta_{j+1}}$ on an upper edge of a backbone in $\Sigma_c$.
Without loss of generality, we can put $\Lambda_{\mathrm{P}}^{-1}$'s on the left hand side of 
the $M$'s.
Indices $\alpha_j,\beta_j(=1,\ldots,N)$ are assigned to two upright line segments nipping a marked point
in the upper edge of the backbone, see Figure \ref{rule_L}.
\item[L3] 
If two matrix elements $U_{\alpha_{j}\beta_{j}}$ and $V_{\alpha_{j+1}\beta_{j+1}}$ ($U,V=M$ or $\Lambda_{\mathrm{L}}^{-1}$) on the same backbone are adjacent, we form a matrix product $(UV)_{\alpha_{j}\beta_{j+1}}$.
This graphical rule is the same as the rule \textbf{C2}.

\item[L4] 
If a matrix product  
\begin{align}
(\Lambda_{\mathrm{L}}^{-1}M)^i_{\alpha_1\beta_i}
\nonumber
\end{align}
corresponds to a backbone with a marked point,
we assign the expression
\begin{align}
N\sum_{\alpha_1,\beta_i=1}^N(\Lambda_{\mathrm{L}}^{-1})_{\alpha_1\beta_i}
\nonumber
\end{align}
to the bottom edge of this backbone. This gives the contribution
\begin{align}
N\mathrm{Tr}((\Lambda_{\mathrm{L}}^{-1}M)^i\Lambda_{\mathrm{L}}^{-1})
\nonumber
\end{align}
with  $i$ chord ends and therefore $i+1$ new marked points.
\item[L5] The Wick contraction between $M_{\alpha_j\beta_j}$ and $M_{\alpha'_{j'}\beta'_{j'}}$ corresponds to a band connecting two chord ends.
This graphical rule is the same as the rule \textbf{C4}.
\end{description}
\begin{figure}[h]
\begin{center}
   \includegraphics[width=120mm,clip]{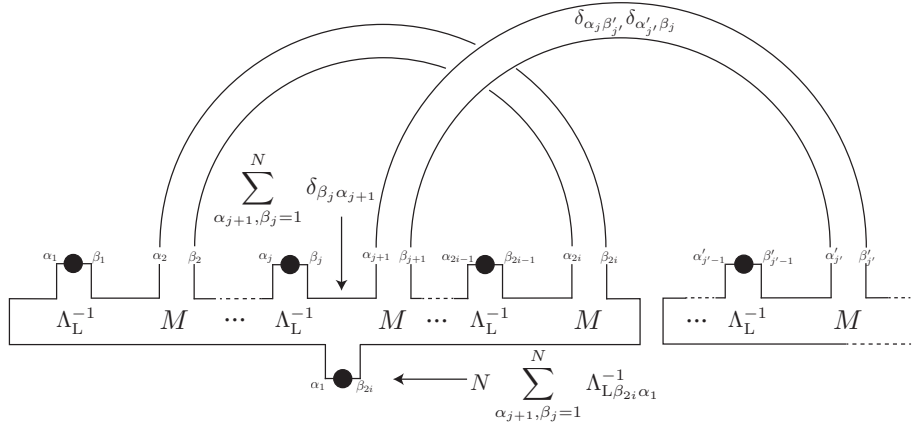}
\end{center}
\caption{\label{rule_L} Bijective correspondence between decorated chord diagrams and matchings of Wick contractions.}
\end{figure}

Repeating the same discussions as in the previous subsection, 
one finds that every chord diagram with the backbone spectrum $\{b_i\}$
corresponds to matchings with $k=\sum ib_i/2$ Wick contractions, which arise from the following Gaussian average
\begin{align}
{W}_N^{\mathrm{L}}(\{b_i\};\{q_i\})
=\Bigl\langle
\prod_{i\ge 0}\left(
N\mathrm{Tr}(\Lambda_{\mathrm{L}}^{-1}M)^{i}\Lambda_{\mathrm{L}}^{-1}
\right)^{b_{i}}
\Bigr\rangle_N^{\mathrm{G}},
\label{GaussL}
\end{align}
where we introduced Miwa times
\begin{align}
q_i=\frac{1}{N}{\mathrm{Tr}}\Lambda_{\mathrm{L}}^{-i}.
\label{Miwa2}
\end{align} 

\begin{figure}[h]
\begin{center}
   \includegraphics[width=110mm,clip]{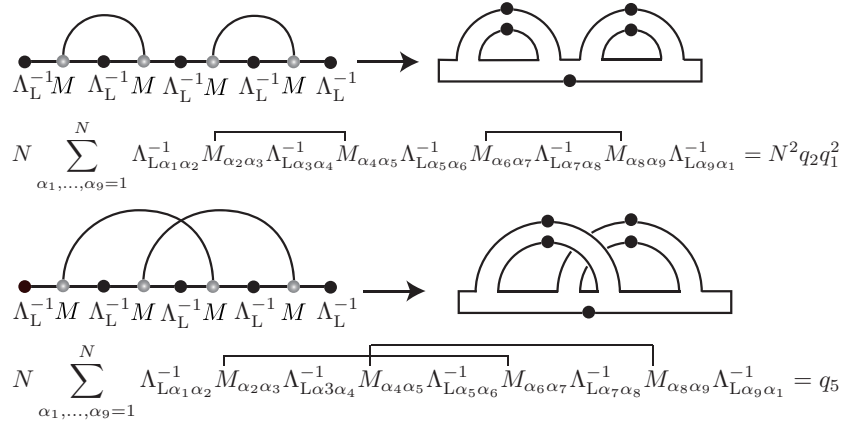}
\end{center}
\caption{\label{fat_length_ext22}  Chord diagrams of types $\{g,k;\{b_i\};\{p_i\}\}=\{0,2;b_5=1;p_1=2,p_3=1\}$ and $\{g,k;\{b_i\};\{p_i\}\}=\{1,2;b_5=1;p_5=1\}$.}
\end{figure}

It follows from the rules \textbf{L1}--\textbf{L5} that $i$ $\Lambda_{\mathrm{L}}^{-1}$'s are aligned along the boundary cycle with length $i$.
Therefore, for chord diagrams with the backbone spectrum $\{b_i\}$ and the boundary length spectrum $\{p_i\}$,
the corresponding Wick contractions in ${W}_N^{\mathrm{L}}(\{b_i\};\{q_i\})$
involve the factor
\begin{align}
N^{b-k+n}\prod_{i\ge 1}q_i^{p_i},
\nonumber
\end{align}
see Figure \ref{fat_length_ext22}. The key proposition of this subsection follows.
\begin{proposition}\label{prop:length_prop1}
The Gaussian average ${W}_N^{\mathrm{L}}(\{b_i\};\{q_i\})$ in eq.(\ref{GaussL})
is the generating function of the numbers $\widehat{N}_{k,b}(\{b_i\},\{p_i\})$ of chord diagrams with the backbone spectrum $\{b_i\}$
\begin{align}
{W}_N^{\mathrm{L}}(\{b_i\};\{q_i\})=\sum_{\{p_i\}}\widehat{N}_{k,b}(\{b_i\},\{p_i\})
N^{b-k+n}\prod_{i\ge 1}q_i^{p_i}.
\end{align}
\end{proposition}

We also consider the full generating function for the numbers $\widehat{N}_{k,b}(\{b_i\},\{p_i\})$ of chord diagrams
\begin{align}
Z_N^{\mathrm{L}}(y;\{s_i\};\{q_i\})
=\sum_{\{b_i\}}\sum_{\{p_i\}}\widehat{N}_{k,b}(\{b_i\},\{p_i\})
N^{b-k+n}y^k\prod_{i\ge 0}s_i^{b_i}\prod_{i\ge 1}q_i^{p_i}.
\nonumber
\end{align}
This full generating function is given by the sum of Gaussian averages (\ref{GaussL}),
and in consequence by the following Hermitian matrix integral
\begin{align}
\begin{split}
&
Z_N^{\mathrm{L}}(y;\{s_i\};\{q_i\})
=\sum_{\{b_i\}}\frac{1}{\prod_ib_i!}y^kW_N^{\mathrm{L}}(\{b_i\},\{y^{-ib_i/2}q_i\})\prod_is_i^{b_i}
\\
&
=\frac{1}{\mathrm{Vol}_N}\int_{{\mathcal H}_N} dM\;\exp\bigg[
-N\mathrm{Tr}\bigg(
\frac{M^2}{2}-
\sum_{i\ge 0}s_{i}y^{i/2}\left(\Lambda_{\mathrm{L}}^{-1}M\right)^{i}
\Lambda_{\mathrm{L}}^{-1}
\bigg)
\bigg].
\label{Z_bdy_length2}
\end{split}
\end{align}
Comparing this matrix integral and the generating function $Z_N^{\mathrm{L}}(y;\{s_i\};\{q_i\})$ in equation (\ref{partition_length_ori}), we arrive at the main theorem of this subsection.
\begin{theorem}\label{matrix_length_ori}
The matrix integral (\ref{Z_bdy_length2}) agrees with the generating function (\ref{partition_length_ori})
\begin{align}
Z_N^{\mathrm{L}}(y;\{s_i\};\{q_i\})=Z^{\mathrm{L}}(N^{-1},y;\{s_i\};\{q_i\}).
\end{align}
\end{theorem}

\subsubsection*{\bfseries Specialization of the model}

The cut-and-join equation for the numbers of chord diagrams is discussed in Subsection \ref{subsec:bl_h}. For technical reasons, the partial differential equation for the generating function (\ref{partition_length_ori}) with general parameter $\{s_i\}$ cannot be written in a simple form.
Therefore we consider the specialization of the generating function (\ref{partition_length_ori}) defined by\footnote{
In \cite{AAPZ}, the length spectrum generating function $G_b(x,y;\{s_i\})$
is the same as in this specialized model.}
\begin{equation}
s_i=s.
\nonumber
\end{equation}
Under this specialization, the matrix integral (\ref{Z_bdy_length2}) reduces to
\begin{align}
\begin{split}
&
Z_N^{\mathrm{L}}(y;s;\{q_i\})=Z_N^{\mathrm{L}}(y;\{s_i=s\};\{q_i\})
\\
&
=\frac{1}{\mathrm{Vol}_N}\int_{{\mathcal H}_N} dM\;\exp\bigg[
-N\mathrm{Tr}\bigg(\frac{M^2}{2}
-\frac{s}{1-y^{1/2}\Lambda_{\mathrm{L}}^{-1}M}
\Lambda_{\mathrm{L}}^{-1}\bigg)
\bigg].
\label{Z_length}
\end{split}
\end{align}
For $Z^{\mathrm{L}}(x,y;s;\{q_i\})=Z^{\mathrm{L}}(x,y;\{s_i=s\};\{q_i\})$, we find
\begin{align}
Z_N^{\mathrm{L}}(y;s;\{q_i\})=Z^{\mathrm{L}}(N^{-1},y;s;\{q_i\}).
\end{align}
In Subsection \ref{subsec:bl_h} we derive the cut-and-join equation for this specialized model, and show the agreement with the cut-and-join equation found by combinatorial means in \cite{AAPZ}.

\subsection{\label{subsec:bpl}The boundary length and point spectrum and the unified model}

So far we have discussed separately the enumeration of chord diagrams and partial chord diagrams labeled by the boundary point spectrum and the boundary length spectrum. In this subsection we consider a unification of these two kinds of spectra, which is referred to as the boundary length and point spectrum. This unified spectrum was introduced and analyzed by cut-and-join methods in \cite{Andersen_new}. In what follows we construct a matrix model that encodes this new spectrum, and in Subsection \ref{subsec:bpl_h} we show how the cut-and-join equation found in \cite{Andersen_new} follows from this matrix model.

\begin{figure}[h]
\begin{center}
   \includegraphics[width=120mm,clip]{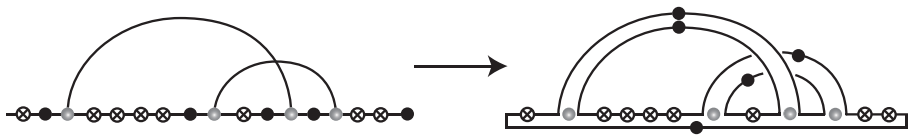}
\end{center}
\caption{\label{fat_pl} Decorating a partial chord diagram with the boundary label $\ii=(1,0,1,4,2)$ with marked points for partitions.}
\end{figure}

The boundary length and point spectrum $\{n_{\ii}\}$ is defined as follows \cite{Andersen_new}.
\begin{definition}
Let $c$ be a partial chord diagram. We associate the $K$ tuple of numbers $\ii=(i_1,i_2,\ldots,i_K)$ to a boundary component of $\Sigma_c$,  if we find the tuple $\ii$ of marked points around this boundary component, once we record different numbers of marked points in between chord ends and backbone underpasses along the boundary in the cyclic order induced from the orientation of $\Sigma_c$.
The boundary length and point spectrum $\{n_{\ii}\}$ counts the number of boundary cycles indexed by $\ii$ for the partial chord diagram $c$.
\end{definition}

To enumerate the number of partial chord diagrams labeled by $\{g,k,l;\{b_i\};\{n_{\ii}\}\}$,
we consider the generating functions introduced in \cite{Andersen_new}.
\begin{definition}
Let ${\mathcal N}_{g,k,l}(\{b_i\},\{n_{\ii}\})$ denote the number of connected chord diagrams labeled by the set of parameters $(g,k,l;\{b_i\};\{n_{\ii}\})$ in the boundary length and point spectrum. The generating function for these numbers is defined as
\begin{align}
\begin{split}
&
{\mathcal F}(x,y;\{s_i\};\{u_{\ii}\})=\sum_{b\ge 1}{\mathcal F}_b(x,y;\{s_i\};\{u_{\ii}\}),
\\
&
{\mathcal F}_b(x,y;\{s_i\};\{u_{\ii}\})=\frac{1}{b!}\sum_{\sum_ib_i=b}
\sum_{\{n_{\ii}\}}
{\mathcal N}_{g,k,l}(\{b_i\},\{n_{\ii}\})x^{2g-2}y^k
\prod_{i\ge 0}s_i^{b_i}\prod_{K\ge 1}\prod_{\{i_L\}_{L=1}^K} u_{\ii}^{n_{\ii}}.
\label{F_blp}
\end{split}
\end{align}
Exponentiating this generating function, one obtains the full generating function for the numbers $\widehat{\mathcal{N}}_{k,b,l}(\{b_i\},\{n_{\ii}\})$ of partial chord diagrams
\begin{align}
&
{\mathcal Z}(x,y;\{s_i\};\{u_{\ii}\})=\exp\left[{\mathcal F}(x,y;\{s_i\};\{u_{\ii}\})\right]
\nonumber \\
&
=\sum_{\{b_i\}}\sum_{\{n_{\ii}\}}
\widehat{\mathcal{N}}_{k,b,l}(\{b_i\},\{n_{\ii}\})
x^{-b+k-n}y^k\prod_{i\ge 0}s_i^{b_i}\prod_{K\ge 1}\prod_{\{i_L\}_{L=1}^K} u_{\ii}^{n_{\ii}},
\label{Z_blp}
\end{align}
where $l$, $k$, and $b$ obey
\begin{align}
l=\sum_{K\ge 1}\sum_{\{i_L\}_{L=1}^K}\sum_{L=1}^Ki_Ln_{(i_1,\ldots,i_K)},\qquad
2k+l=\sum_{i\ge 1}ib_i, \qquad b=\sum_{i\ge 0}b_i.
\nonumber
\end{align}
\end{definition}

The enumeration of partial chord diagrams decorated by the boundary length and point spectrum can also be expressed in terms of Gaussian averages over Hermitian matrices. To this end we again make use of extra marked points, just as in the previous section (concerning the length spectrum to mark the separation between marked points on the backbone, counted by the index $\ii$), see Figure \ref{fat_pl}.
Indeed, the boundary length and point spectrum also encodes the length spectrum, simply as the number $K$ of partitions of marked points on boundary cycles. 

To represent the boundary length and point spectrum,  we introduce two external matrices
$\Lambda_{\mathrm{P}}$ and $\Lambda_{\mathrm{L}}$.
In order to faithfully represent the ordering between marked points and partitions on each boundary cycle,
we assume that these two external matrices do not commute
\begin{align}
[\Lambda_{\mathrm{P}},\Lambda_{\mathrm{L}}]\ne 0.
\nonumber
\end{align}
The correspondence between partial chord diagrams with the backbone spectrum $\{b_i\}$ and matchings by Wick contractions in a Gaussian average
is given by a combination of the previous rules \textbf{C1}, \textbf{C2}, \textbf{P2}, \textbf{L2},  \textbf{P4}, \textbf{L4}, and \textbf{L5}.
We summarize this correspondence in Table \ref{table:chord_wick}.

\begin{table}[htb]
\caption{\label{table:chord_wick}The correspondence between partial chord diagrams with the backbone spectrum $\{b_i\}$ and matchings by Wick contractions in the Gaussian average.}
  \begin{tabular}{|c|c|} \hline
    A partial chord diagram &  Gaussian average  \\ \hline \hline
    A chord end on a backbone & $\Lambda_{\mathrm{L}}^{-1}M$ \\
    A marked point on a backbone & $\Lambda_{\mathrm{P}}$ \\
    An underside of a backbone & $N\Lambda_{\mathrm{L}}^{-1}$ \\
    A backbone & $N{\rm Tr}\left(\Lambda_{\mathrm{L}}^{-1}
\Lambda_{\mathrm{P}}^{\alpha_1}\Lambda_{\mathrm{L}}^{-1}\Lambda_{\mathrm{P}}^{\alpha_2}\Lambda_{\mathrm{L}}^{-1}\cdots\Lambda_{\mathrm{P}}^{\alpha_K}\Lambda_{\mathrm{L}}^{-1}
\right)$ \\
    A chord & Wick contraction $\contraction[1ex]{}{X}{X}{}MM$ \\
 \hline
  \end{tabular}
\end{table}

Based on these rules, one finds a bijective correspondence between partial chord diagrams with the backbone spectrum $\{b_i\}$
and matchings by Wick contractions in the Gaussian average  
\begin{align}
{\mathcal W}_N(\{b_i\};\{u_{\ii}\})=\Bigl\langle
\prod_{i\ge 0}\left(
N\mathrm{Tr}(\Lambda_{\mathrm{L}}^{-1}M+\Lambda_{\mathrm{P}})^{i}\Lambda_{\mathrm{L}}^{-1}
\right)^{b_{i}}
\Bigr\rangle_N^{\mathrm{G}},
\label{GaussPL}
\end{align}
where in order to represent trace factors $\Lambda_{\mathrm{P}}$ and $\Lambda_{\mathrm{L}}$ we introduced the \textit{generalized Miwa times}
\begin{align}
u_{(i_1,\ldots,i_K)}=\frac{1}{N}{\mathrm{Tr}}\left(
\Lambda_{\mathrm{P}}^{i_1}\Lambda_{\mathrm{L}}^{-1}\Lambda_{\mathrm{P}}^{i_2}\Lambda_{\mathrm{L}}^{-1}\cdots\Lambda_{\mathrm{P}}^{i_K}\Lambda_{\mathrm{L}}^{-1}
\right).
\label{blp_time}
\end{align}
If a partial chord diagram $c$ contains a boundary cycle labeled by $\ii=(i_1,\ldots,i_K)$,
one finds the following trace factor in the corresponding Wick contractions in ${\mathcal W}_N(\{b_i\};\{u_{\ii}\})$
\begin{align}
{\mathrm{Tr}}\left(
\Lambda_{\mathrm{P}}^{i_1}\Lambda_{\mathrm{L}}^{-1}\Lambda_{\mathrm{P}}^{i_2}\Lambda_{\mathrm{L}}^{-1}\cdots\Lambda_{\mathrm{P}}^{i_K}\Lambda_{\mathrm{L}}^{-1}
\right).
\nonumber
\end{align}
Finally, combining Propositions \ref{prop:iI_N} and \ref{prop:length_prop1}, we obtain the key proposition.
\begin{proposition}\label{prop:blp}
The Gaussian average ${\mathcal W}_N(\{b_i\};\{u_{\ii}\})$ in the equation (\ref{GaussPL})
is the generating function for the numbers $\widehat{\mathcal{N}}_{k,b,l}(\{b_i\},\{n_{\ii}\})$
of partial chord diagrams 
\begin{align}
{\mathcal W}_N(\{b_i\};\{u_{\ii}\})=
\sum_{\{n_{\ii}\}}
\widehat{\mathcal N}_{k,b,l}({\{b_i\},\{n_{\ii}\}})N^{b-k+n}
\prod_{K\ge 1}\prod_{\{i_L\}_{L=1}^K} u_{\ii}^{n_{\ii}}.
\end{align}
\end{proposition}

Repeating the same combinatorics as in the previous subsections, 
we find the main theorem of this section.

\begin{theorem}\label{Thm:blp_matrix}
The Hermitian matrix integral 
\begin{align}
\begin{split}
&
{\mathcal Z}_N(y;\{s_i\};\{u_{\ii}\}) =
\\
&
=\frac{1}{{\mathrm{Vol}}_N}\int dM\;\exp\bigg[
-N{\mathrm{Tr}}\bigg(\frac{M^2}{2}-\sum_{i\ge 0}s_i(y^{1/2}\Lambda_{\mathrm{L}}^{-1}M+\Lambda_{\mathrm{P}})^i\Lambda_{\mathrm{L}}^{-1}\bigg)
\bigg]
\label{blp_model}
\end{split}
\end{align}
agrees with the generating function (\ref{Z_blp})
\begin{align}
{\mathcal Z}_N(y;\{s_i\};\{u_{\ii}\})={\mathcal Z}(N^{-1},y;\{s_i\};\{u_{\ii}\}).
\label{Z_blp1}
\end{align}
\end{theorem}

\section{\label{sec:mat_heat}Cut-and-join equations via matrix models}

In Section \ref{sec:const_mat} we discussed matrix models that enumerate partial chord diagrams filtered by the boundary point spectrum, the boundary length spectrum, and the boundary length and point spectrum. In this section we derive partial differential equations for these matrix models, and show that they agree with the cut-and-join equations found in \cite{AAPZ, Andersen_new}. To derive these differential equations, it is useful to introduce the following matrix integral.

\begin{definition}\label{def:master_mat}
Let $A$ and $B$ denote invertible matrices of rank $N$. We define a formal matrix integral with parameters $y$, $\{g_i\}_{i=-\infty}^{+\infty}$, and matrices $A$ and $B$, as follows
\begin{align}
\begin{split}
&
Z_N(y;\{g_i\};A;B) =
\\
&
=\frac{1}{\mathrm{Vol}_N}\int_{\mathcal{H}_N} dM\;
\exp\bigg[-N\mathrm{Tr}\bigg(\frac12M^2-\sum_{i\in {\IZ}}g_i(y^{1/2}B^{-1}M+A)^iB^{-1}\bigg)\bigg].
\label{gblp_def}
\end{split}
\end{align}
\end{definition}

By the following specializations of this matrix integral one finds matrix integrals discussed in Section \ref{sec:const_mat}
\begin{align}
\begin{split}
&
Z_N^{\mathrm{P}}(y;\{s_i\};\{r_i\}):\ \
g_{i<0}=0,\ \ \ g_{i\ge 0}=s_i,\ \ \
A=\Lambda_{\mathrm{P}},\ \ \ B=I_N,
\\
&
Z_N^{\mathrm{L}}(y;\{s_i\};\{q_i\}):\ \
g_{i<0}=0,\ \ \ g_{i\ge 0}=s_i,\ \ \
A=0,\ \ \ B=\Lambda_{\mathrm{L}},
\\
&
Z_N^{\mathrm{L}}(y;s;\{q_i\}):\ \
g_{i\neq -1}=0,\ \ \
g_{-1}=-s,\ \ \
A=\Lambda_{\mathrm{L}},\ \ \ B=-I_N,
\\
&
\mathcal{Z}_N(y;\{s_i\};\{u_{\ii}\}):\ \
g_{i<0}=0,\ \ \ g_{i\ge 0}=s_i,\ \ \
A=\Lambda_{\mathrm{P}},\ \ \ B=\Lambda_{\mathrm{L}},
\nonumber
\end{split}
\end{align}
where $I_N$ is the rank $N$ identity matrix.

The matrix integral (\ref{gblp_def}) satisfies the following partial differential equation.

\begin{proposition}\label{prop:master_heat}
The matrix integral $Z_N(y;\{g_i\};A;B)$ obeys a partial differential equation
\begin{equation}
\bigg[\frac{\partial}{\partial y}
-\frac{1}{2N}\mathrm{Tr}(B^{-1})^{\mathrm{T}}\frac{\partial}{\partial A}
(B^{-1})^{\mathrm{T}}\frac{\partial}{\partial A}\bigg]Z_N(y;\{g_i\};A;B)=0,
\label{gblp_heat}
\end{equation}
where the trace in the second term is defined, for rank $N$ matrices $X$ and $Y$, as
\begin{equation}
\mathrm{Tr}X\frac{\partial}{\partial Y}=
\sum_{\alpha,\beta=1}^NX_{\alpha\beta}\frac{\partial}{\partial Y_{\beta\alpha}}.
\nonumber
\end{equation}
\end{proposition}

\begin{proof}
By a shift $M=X-y^{-1/2}BA$, the matrix integral (\ref{gblp_def}) can be rewritten as
\begin{align}
&
Z_N(y;\{g_i\};A;B) =
\nonumber\\
&
=\frac{1}{\mathrm{Vol}_N}\int_{\widetilde{\mathcal{H}}_N} dX\;
\exp\bigg[-N\mathrm{Tr}\bigg(\frac12(X-y^{-1/2}BA)^2
-\sum_{i\in {\IZ}}y^{i/2}g_i(B^{-1}X)^iB^{-1}\bigg)\bigg].
\label{gblp2_def}
\end{align}
Here $\widetilde{\mathcal{H}}_N$ is the space of shifted matrices $X=M+y^{-1/2}BA$ with $M \in \mathcal{H}_N$.
The invariance of this matrix integral under the infinitesimal scaling $X_{\alpha\beta}\to (1+\epsilon)X_{\alpha\beta}$ leads to a constraint equation
\begin{equation}
\bigg<N^2-N\mathrm{Tr}X^2+y^{-1/2}N\mathrm{Tr}BAX
+N\sum_{i\in{\IZ}}iy^{i/2}g_i\mathrm{Tr}(B^{-1}X)^iB^{-1}\bigg>=0,
\label{gblp_sv0}
\end{equation}
where the first term $N^2$ comes from the measure factor, as $dX \to (1+N^2\epsilon)dX$. Here we have defined the unnormalized average for an observable $\mathcal{O}(X)$ 
\begin{equation}
\left<\mathcal{O}(X)\right>=
\int_{\widetilde{\mathcal{H}}_N} dX\; \mathcal{O}(X)
\exp\bigg[-N\mathrm{Tr}\bigg(\frac12(X-y^{-1/2}BA)^2-\sum_{i\in {\IZ}}y^{i/2}g_i(B^{-1}X)^iB^{-1}\bigg)\bigg].
\nonumber
\end{equation}
Using
\begin{align}
\begin{split}
&
\frac{1}{N}y^{1/2}\sum_{\gamma=1}^N(B^{-1})^{\mathrm{T}}_{\beta\gamma}\frac{\partial}{\partial A_{\gamma\alpha}}Z_N(y;\{g_i\};A;B)
=\left<X_{\alpha\beta}-y^{-1/2}N\sum_{\gamma=1}^NB_{\alpha\gamma}A_{\gamma\beta}\right>,\\
&
\frac{1}{N}\frac{\partial}{\partial g_i}Z_N(y;\{g_i\};A;B)
=y^{i/2}\left<\mathrm{Tr}(B^{-1}X)^iB^{-1}\right>,
\nonumber
\end{split}
\end{align}
one finds that the constraint equation (\ref{gblp_sv0}) yields
\begin{equation}
\bigg[
-\frac{1}{N}y\mathrm{Tr}(B^{-1})^{\mathrm{T}}\frac{\partial}{\partial A}(B^{-1})^{\mathrm{T}}\frac{\partial}{\partial A}
-\mathrm{Tr}A^{\mathrm{T}}\frac{\partial}{\partial A}+
\sum_{i\in{\IZ}}ig_i\frac{\partial}{\partial g_i}\bigg]
Z_N(y;\{g_i\};A;B)=0.
\nonumber
\end{equation}
It follows from (\ref{gblp2_def}) that the last two derivatives in the expression above can be replaced by $2y\partial/\partial y$, so that the partial differential equation (\ref{gblp_heat}) is obtained.
\end{proof}

\begin{remark}
In the above proof of the constraint equation (\ref{gblp_sv0}) we considered the infinitesimal scaling $X_{\alpha\beta}\to (1+\epsilon)X_{\alpha\beta}$. More generally, matrix integral (\ref{gblp2_def}) is invariant under infinitesimal shifts
\begin{align}
X_{\alpha\beta}\ \longrightarrow\ X_{\alpha\beta}+\epsilon (X^{n+1})_{\alpha\beta},\ \ \ \
n=-1,0,1,\ldots.
\nonumber
\end{align}
It is known that for the matrix integral without external matrices $A$ and $B$ this symmetry yields the Virasoro symmetry, and in particular the scaling $X_{\alpha\beta}\to (1+\epsilon)X_{\alpha\beta}$ is related to the Virasoro generator $L^{\mathrm{Vir}}_0$ \cite{Fukuma:1990jw, Dijkgraaf:1990rs}.\footnote{
In \cite{MY,CY}, the Schwinger-Dyson approach to the enumeration of chord diagrams is also discussed.
}
\end{remark}

\subsection{\label{subsec:bp_h}The boundary point spectrum}

In Subsection \ref{subsec:bp} we showed that the matrix integral $Z_N^{\mathrm{P}}(y;\{s_i\};\{r_i\})$ in (\ref{partition_fn_pt})
\begin{align}
\begin{split}
&
Z_N^{\mathrm{P}}(y;\{s_i\};\{r_i\}) =
\\
&
=\frac{1}{\mathrm{Vol}_N}\int_{{\mathcal H}_N} dM\;\exp\bigg[
-N\mathrm{Tr}\bigg(\frac{M^2}{2}
-\sum_{i\ge 0}s_iy^{i/2}(M+y^{-1/2}\Lambda_{\mathrm{P}})^i\bigg)\bigg],
\nonumber
\end{split}
\end{align}
enumerates partial chord diagrams labeled by the boundary point spectrum. 
By the specialization 
\begin{align}
g_{i<0}=0,\ \ \ \ g_{i\ge 0}=s_i,\ \ \ \
A=\Lambda_{\mathrm{P}},\ \ \ \ B=I_N=\textrm{identity matrix},
\nonumber
\end{align}
of the matrix integral $Z_N(y;\{g_i\};A;B)$ in (\ref{gblp_def}) we see that
\begin{equation}
Z_N(y;s_{i<0}=0, \{s_i\}_{i\ge 0};\Lambda_{\mathrm{P}};I_N)=Z_N^{\mathrm{P}}(y;\{s_i\};\{r_i\}),
\end{equation}
where the reverse Miwa times $r_i$ are defined in (\ref{Miwa1}). From (\ref{gblp_heat}) we obtain the partial differential equation satisfied by $Z_N^{\mathrm{P}}(y;\{s_i\};\{r_i\})$.

\begin{corollary}\label{cor:point_heat_ori}
The matrix integral $Z_N^{\mathrm{P}}(y;\{s_i\};\{r_i\})$ obeys the partial differential equation
\begin{equation}
\bigg[\frac{\partial}{\partial y}
-\frac{1}{2N}\mathrm{Tr}\frac{\partial^2}{\partial \Lambda_{\mathrm{P}}^2}\bigg]Z_N^{\mathrm{P}}(y;\{s_i\};\{r_i\})=0.
\label{bp_heat}
\end{equation}
\end{corollary}

This corollary implies the following theorem.

\begin{theorem}\label{thm:cut_and_join_pt_ori}
Let $L_0$ and $L_2$ be the differential operators\footnote{In \cite{Mor_Sha} the differential operators $L_0$ and $L_2$ were denoted by $W^{(3)}$.}
\begin{align}
\begin{split}
&
L_0=\frac{1}{2}\sum_{i\ge 2}\sum_{j=0}^{i-2}ir_jr_{i-j-2}\frac{\partial}{\partial r_{i}},
\\
&
L_2=\frac{1}{2}\sum_{i\ge 2}\sum_{j=1}^{i-1}j(i-j)r_{i-2}\frac{\partial^2}{\partial r_i\partial r_{i-j}}.
\label{L02_pt}
\end{split}
\end{align}
The matrix integral $Z_N^{\mathrm{P}}(y;\{s_i\};\{r_i\})$ obeys the cut-and-join equation
\begin{align}
\frac{\partial}{\partial y}Z_N^{\mathrm{P}}(y;\{s_i\};\{r_i\})={\mathcal L}Z_N^{\mathrm{P}}(y;\{s_i\};\{r_i\}),
\label{cut_and_join_pt_ori}
\end{align}
where
\begin{align}
{\mathcal L}=L_0+\frac{1}{N^2}L_2.
\nonumber
\end{align}
The formal solution of this cut-and-join equation, which gives the matrix integral $Z_N^{\mathrm{P}}(y;\{s_i\};\{r_i\})$, is iteratively determined from the initial condition at $y=0$,
\begin{align}
Z_N^{\mathrm{P}}(y;\{s_i\};\{r_i\})=\mathrm{e}^{y\mathcal{L}}
Z_N^{\mathrm{P}}(0;\{s_i\};\{r_i\})=
\mathrm{e}^{y\mathcal{L}}\mathrm{e}^{N^2\sum_{i\ge 0}s_ir_i}.
\label{initial_pt_ori}
\end{align}
\end{theorem}

This theorem follows from the lemma below by rewriting the derivative  ${\mathrm Tr}\partial^2/\partial \Lambda_{\mathrm{P}}^2$ in the partial differential equation (\ref{bp_heat}).

\begin{lemma}\label{lem:miwa_deriv}
For a function $f(\{r_i\})$ of the reverse Miwa times $r_i$,
the derivative ${\mathrm Tr}\partial^2/\partial \Lambda_{\mathrm{P}}^2$ can be rewritten as
\begin{align}
\frac{1}{2N}\mathrm{Tr}\frac{\partial^2}{\partial \Lambda_{\mathrm{P}}^2}f(\{r_i\})
=\bigg(L_0+\frac{1}{N^2}L_2\bigg)f(\{r_i\}).
\label{miwa_deriv_bp}
\end{align}
\end{lemma}
\begin{proof}
Consider the derivative $\partial/\partial\Lambda_{\mathrm{P}\beta\alpha}$ of $r_i$,
\begin{align}
\frac{\partial r_i}{\partial \Lambda_{\mathrm{P}\beta\alpha}}
=\frac{i}{N}\Lambda^{i-1}_{\mathrm{P}\alpha\beta},\qquad
\mathrm{Tr}\frac{\partial^2 r_i}{\partial\Lambda_{\mathrm{P}}^2}
=iN\sum_{j=0}^{i-2}r_jr_{i-j-2}.
\nonumber
\end{align}
Then the derivative $\mathrm{Tr}\partial^2/\partial \Lambda_{\mathrm{P}}^2$ of the function $f(\{r_i\})$ is re-expressed as
\begin{align}
\frac{1}{2N}\mathrm{Tr}\frac{\partial^2}{\partial \Lambda_{\mathrm{P}}^2}f(\{r_i\})
&=
\frac{1}{2N}\sum_{i\ge 0}
\mathrm{Tr}\frac{\partial^2
r_i}{\partial\Lambda_{\mathrm{P}}^2}\frac{\partial
f(\{r_i\})}{\partial r_i}
+\frac{1}{2N}\sum_{i,j\ge 0}\sum_{\alpha,\beta=1}^N\frac{\partial r_i}{\partial
\Lambda_{\mathrm{P}\beta\alpha}}\frac{\partial r_j}{\partial \Lambda_{\mathrm{P}\alpha\beta}}\frac{\partial^2
f(\{r_i\})}{\partial r_i\partial r_j}
\nonumber\\
&
=\frac{1}{2}\sum_{i\ge 2}\sum_{j=0}^{i-2}ir_jr_{i-j-2}\frac{\partial
 f(\{r_i\})}{\partial r_i}
+\frac{1}{2N^2}\sum_{i,j\ge 1}ijr_{i+j-2}\frac{\partial^2f(\{r_i\})}
{\partial r_i\partial r_j}.
\nonumber
\end{align}
This coincides with the right hand side of (\ref{miwa_deriv_bp}).
\end{proof}

The cut-and-join equation for the rescaled matrix integral (\ref{rescale_pt}) yields
\begin{align}
\frac{\partial}{\partial y}Z_{t_0 N}^{\mathrm{P}}(t_0 y;\{t_0^{-1}s_i\};\{t_0^{-1}t_i\})={\mathcal L}Z_{t_0 N}^{\mathrm{P}}(t_0 y;\{t_0^{-1}s_i\};\{t_0^{-1}t_i\}),
\end{align}
where ${\mathcal L}$ is given by
\begin{align}
\begin{split}
&
{\mathcal L}=L_0+x^2L_2, \quad x=N^{-1},\\
&
L_0=\frac{1}{2}\sum_{i\ge 2}\sum_{j=0}^{i-2}it_jt_{i-j-2}\frac{\partial}{\partial t_{i}},
\quad
L_2=\frac{1}{2}\sum_{i\ge 2}\sum_{j=1}^{i-1}j(i-j)t_{i-2}\frac{\partial^2}{\partial t_i\partial t_{i-j}}.
\end{split}
\end{align}
This cut-and-join equation agrees with the partial differential equation in Theorem 1 of \cite {AAPZ}, where it was proven combinatorially by the recursion relation for the number of partial chord diagrams.\footnote{For the Grothendieck's dessin counting, a similar cut-and-join equation was found in \cite{Kaz_Zog}.}
This  
completes the proof of Theorem \ref{thm:point_ori}.

\subsection{\label{subsec:bl_h}The boundary length spectrum}

In Subsection \ref{subsec:bl} we showed that the matrix integral $Z_N^{\mathrm{L}}(y;\{s_i\};\{q_i\})$ in (\ref{Z_bdy_length2}) enumerates chord diagrams labeled by the boundary length spectrum.
By the specialization 
\begin{align}
g_{i<0}=0,\ \ \ \ g_{i\ge 0}=s_i,\ \ \ \
A=0,\ \ \ \ B=\Lambda_{\mathrm{L}},
\nonumber
\end{align}
of the matrix integral $Z_N(y;\{g_i\};A;B)$ in (\ref{gblp_def}) we see that
\begin{equation}
Z_N(y;s_{i<0}=0,\{s_i\}_{i\ge 0};0;\Lambda_{\mathrm{L}})=Z_N^{\mathrm{L}}(y;\{s_i\};\{q_i\}),
\end{equation}
where the Miwa times $q_i$ are defined in equation (\ref{Miwa2}).

Obviously, for $A=0$ the partial differential equation (\ref{gblp_heat}) does not hold. Instead we consider the matrix integral (\ref{Z_length}) obtained by the specialization $s_i=s$
\begin{align}
Z_N^{\mathrm{L}}(y;s;\{q_i\})
=\frac{1}{\mathrm{Vol}_N}\int_{{\mathcal H}_N} dM\;\exp\bigg[
-N\mathrm{Tr}\bigg(\frac{M^2}{2}
+\frac{s}{y^{1/2}M-\Lambda_{\mathrm{L}}}\bigg)\bigg].
\nonumber
\end{align}
The same matrix integral can be obtained by the specialization 
\begin{align}
g_{i\neq -1}=0,\ \ \ \
g_{-1}=-s,\ \ \ \
A=\Lambda_{\mathrm{L}},\ \ \ \ B=-I_N,
\nonumber
\end{align}
and thus
\begin{equation}
Z_N(y;s_i=-\delta_{i,-1};\Lambda_{\mathrm{L}};B=-I_N)=Z_N^{\mathrm{L}}(y;s;\{q_i\}).
\end{equation}
Then from (\ref{gblp_heat}) we obtain a partial differential equation for $Z_N^{\mathrm{L}}(y;s;\{q_i\})$.

\begin{corollary}\label{cor:length_heat_ori}
The matrix integral $Z_N^{\mathrm{L}}(y;s;\{q_i\})$ obeys the partial differential equation
\begin{equation}
\bigg[\frac{\partial}{\partial y}
-\frac{1}{2N}\mathrm{Tr}\frac{\partial^2}{\partial \Lambda_{\mathrm{L}}^2}\bigg]Z_N^{\mathrm{L}}(y;s;\{r_i\})=0.
\label{bl_heat}
\end{equation}
\end{corollary}

This corollary implies the following theorem.

\begin{theorem}\label{thm:cut_join_length}
Let $K_0$ and $K_2$ be the differential operators
\begin{align}
\begin{split}
&
K_0=\frac{1}{2}\sum_{i\ge 3}\sum_{j=1}^{i-1}(i-2)q_jq_{i-j}
\frac{\partial}{\partial q_{i-2}}, 
\\
&
K_2=\frac{1}{2}\sum_{i\ge 2}\sum_{j=1}^{i-1}j(i-j)q_{i+2}\frac{\partial^2}{\partial q_i\partial q_{i-j}}.
\label{K02}
\end{split}
\end{align}
The matrix integral $Z_N^{\mathrm{L}}(y;s;\{q_i\})$ obeys the cut-and-join equation
\begin{align}
\frac{\partial}{\partial y} Z_N^{\mathrm{L}}(y;s;\{q_i\})={\mathcal K}Z_N^{\mathrm{L}}(y;s;\{q_i\}),
\label{cut_join_length}
\end{align}
where
\begin{align}
{\mathcal K}=K_0+\frac{1}{N^2}K_2.
\nonumber
\end{align}
The formal solution of this cut-and-join equation, which gives the matrix integral $Z_N^{\mathrm{L}}(y;s;\{q_i\})$, is iteratively determined from the initial condition  at $y=0$,
\begin{align}
Z_N^{\mathrm{L}}(y;s;\{q_i\})=\mathrm{e}^{y\mathcal{K}}Z_N^{\mathrm{L}}(y=0;s;\{q_i\})
=\mathrm{e}^{y\mathcal{K}}\mathrm{e}^{N^2sq_1}.
\label{initial_length0}
\end{align}
\end{theorem}

The cut-and-join equation (\ref{cut_join_length}) was combinatorially proven in Theorem 2 of \cite{AAPZ} for the generating function $Z^{\mathrm{L}}(x,y;s;\{q_i\})$ in (\ref{free_length_ori}), and thus Theorem \ref{matrix_length_ori} for $s_i=s$ is reproved.

The claim of Theorem \ref{thm:cut_join_length} is proven by rewriting the derivative $\mathrm{Tr}\partial^2/\partial \Lambda_{\mathrm{L}}^2$ in the partial differential equation (\ref{bl_heat}) using the following lemma.

\begin{lemma}\label{lem:heat_Miwa_length}
For a function $g(\{q_i\})$ of the Miwa times $\{q_i\}$, 
the derivative $\mathrm{Tr}\partial^2/\partial \Lambda_{\mathrm{L}}^2$ can be rewritten as follows
\begin{align}
\frac{1}{2N}\mathrm{Tr}\frac{\partial^2}{\partial \Lambda_{\mathrm{L}}^2}g(\{q_i\})
=\left(K_0+\frac{1}{N^2}K_2\right)g(\{q_i\}).
\label{miwa_deriv_bl}
\end{align}
\end{lemma}
\begin{proof}
By acting $\partial/\partial \Lambda_{\mathrm{L}}$ on the Miwa time $q_i$ one obtains
\begin{align}
\frac{\partial q_i}{\partial \Lambda_{\mathrm{L}\alpha\beta}}=-\frac{i}{N}\Lambda^{-i-1}_{\mathrm{L}\beta\alpha},
\qquad
\mathrm{Tr}\frac{\partial^2 q_i}{\partial\Lambda_{\mathrm{L}}^2}
=iN\sum_{j=1}^{i+1}q_jq_{i-j+2}.
\nonumber
\end{align}
Adopting this relation via the chain rule applied to the $\Lambda_{\mathrm{L}}$ derivatives, one finds that 
\begin{align}
\frac{1}{2N}\mathrm{Tr}\frac{\partial^2g(\{q_i\})}{\partial \Lambda_{\mathrm{L}}^2}
&=
\frac{1}{2N}\sum_{i\ge 0}\mathrm{Tr}\frac{\partial^2 q_i}{\partial \Lambda_{\mathrm{L}}^2}
\frac{\partial g(\{q_i\})}{\partial q_i}
+\frac{1}{2N}\sum_{i,j\ge 0}\sum_{\alpha,\beta=1}^N\frac{\partial q_i}{\partial\Lambda_{\mathrm{L}\alpha\beta}}
\frac{\partial
q_j}{\partial\Lambda_{\mathrm{L}\beta\alpha}}\frac{\partial^2g(\{q_i\})}{\partial q_i\partial q_j}
\nonumber\\
&=
\frac{1}{2}\sum_{i\ge 1}iq_jq_{i-j+2}\frac{\partial g(\{q_i\})}{\partial q_i}
+\frac{1}{2N^2}\sum_{i,j\ge 1}ijq_{i+j+2}\frac{\partial^2g(\{q_i\})}{\partial q_i\partial q_j}.
\nonumber
\end{align}
This coincides with the right hand side of (\ref{miwa_deriv_bl}).
\end{proof}

\subsection{\label{subsec:bpl_h}The boundary length and point spectrum}

In Subsection \ref{subsec:bpl} we showed that the matrix integral 
$\mathcal{Z}_N(y;\{s_i\};\{u_{\ii}\})$ in (\ref{Z_blp})
\begin{align}
\begin{split}
&
\mathcal{Z}_N(y;\{s_i\};\{u_{\ii}\}) =
\\
&
=\frac{1}{{\mathrm{Vol}}_N}\int_{\mathcal{H}_N} dM\;\exp\bigg[
-N{\mathrm{Tr}}\bigg(\frac{M^2}{2}-\sum_{i\ge 0}s_i(y^{1/2}\Lambda_{\mathrm{L}}^{-1}M+\Lambda_{\mathrm{P}})^i\Lambda_{\mathrm{L}}^{-1}\bigg)\bigg]
\nonumber
\end{split}
\end{align}
enumerates partial chord diagrams labeled by the boundary length and point spectrum. By the specialization 
\begin{align}
g_{i<0}=0,\ \ \ \ g_{i\ge 0}=s_i,\ \ \ \
A=\Lambda_{\mathrm{P}},\ \ \ \ B=\Lambda_{\mathrm{L}},
\nonumber
\end{align}
of the matrix integral $Z_N(y;\{g_i\};A;B)$ in (\ref{gblp_def}) we see that
\begin{equation}
Z_N(y;s_{i<0}=0,\{s_i\}_{i\ge 0};\Lambda_{\mathrm{P}};\Lambda_{\mathrm{L}})=\mathcal{Z}_N(y;\{s_i\};\{u_{\ii}\}),
\end{equation}
where the generalized Miwa times $u_{(i_1,\ldots,i_K)}$ are defined in (\ref{blp_time})  
\begin{align}
u_{(i_1,\ldots,i_K)}=\frac{1}{N}{\mathrm{Tr}}\left(
\Lambda_{\mathrm{P}}^{i_1}\Lambda_{\mathrm{L}}^{-1}\Lambda_{\mathrm{P}}^{i_2}\Lambda_{\mathrm{L}}^{-1}\cdots\Lambda_{\mathrm{P}}^{i_K}\Lambda_{\mathrm{L}}^{-1}
\right).
\nonumber
\end{align}
From (\ref{gblp_heat}) we obtain a partial differential equation for $\mathcal{Z}_N(y;\{s_i\};\{u_{\ii}\})$.

\begin{corollary}\label{cor:blp_heat_ori}
The matrix integral $\mathcal{Z}_N(y;\{s_i\};\{u_{\ii}\})$ obeys the partial differential equation
\begin{equation}
\bigg[\frac{\partial}{\partial y}
-\frac{1}{2N}\mathrm{Tr}(\Lambda_{\mathrm{L}}^{-1})^{\mathrm{T}}\frac{\partial}{\partial \Lambda_{\mathrm{P}}}
(\Lambda_{\mathrm{L}}^{-1})^{\mathrm{T}}\frac{\partial}{\partial \Lambda_{\mathrm{P}}}\bigg]\mathcal{Z}_N(y;\{s_i\};\{u_{\ii}\})=0.
\label{blp_heat}
\end{equation}
\end{corollary}

This corollary implies the following main theorem of this section.

\begin{theorem}\label{thm:cut_join_blp}
Let $M_0$ and $M_2$ be the following differential operators with respect to parameters $u_{\ii}$ 
\begin{align}
\begin{split}
M_0&=\frac{1}{2}\sum_{K\ge 1}\sum_{\{i_1,\ldots,i_K\}}\sum_{1\le I\ne M\le
 K}\sum_{\ell=0}^{i_I-1}\sum_{m=0}^{i_M-1}
\\
&\ \ \ \
u_{(i_I-\ell-1,i_{I+1},\ldots,i_{M-1},m)}
u_{(i_M-m-1,i_{M+1},\ldots,i_{I-1},\ell)}\frac{\partial}{\partial
u_{(i_1,\ldots,i_K)}}
\\
&\ \ \ \
+\sum_{K\ge 1}\sum_{\{i_1,\ldots,i_K\}}
\sum_{I=0}^K\sum_{\ell+m\le i_I-2}
u_{(\ell,m,i_{I+1},\ldots,i_{I-1})}u_{(i_I-\ell-m-2)}
\frac{\partial}{\partial
u_{(i_1,\ldots,i_K)}},
\\
M_2&=\frac{1}{2}\sum_{K,L\ge 0}\sum_{\{i_1,\ldots,i_K\}}
\sum_{\{j_1,\ldots,j_L\}}
\sum_{I=0}^K\sum_{J=0}^{L}\sum_{\ell=0}^{i_I-1}\sum_{m=0}^{j_J-1}
\\
&\ \ \ \
u_{(i_I-\ell-1,i_{I+1},\ldots,i_{I-1},\ell,j_J-m-1,j_{J+1},\ldots,j_{J-1},m)}
\frac{\partial^2}{\partial u_{(i_1,\ldots,i_K)}\partial u_{(j_1,\ldots,j_{L})}},
\label{M02}
\end{split}
\end{align}
where labels $I,M$'s are defined modulo $K$, and the label $J$ is defined modulo $L$.
The matrix integral ${\mathcal Z}_N(y;\{s_i\};\{u_{{\ii}}\})$ obeys the cut-and-join equation
\begin{align}
\frac{\partial}{\partial y} {\mathcal Z}_N(y;\{s_i\};\{u_{{\ii}}\})
={\mathcal M}{\mathcal Z}_N(y;\{s_i\};\{u_{{\ii}}\}),
\label{cut_join_bpl_ori}
\end{align}
where
\begin{align}
{\mathcal M}=M_0+\frac{1}{N^2}M_2.
\nonumber
\end{align}
The formal solution of this cut-and-join equation, which gives the matrix integral ${\mathcal Z}_N(y;\{s_i\};\{u_{{\ii}}\})$, is iteratively determined from the initial condition at $y=0$,
\begin{align}
{\mathcal Z}_N(y;\{s_i\};\{u_{{\ii}}\})=\mathrm{e}^{y\mathcal{M}}
{\mathcal Z}_N(y=0;\{s_i\};\{u_{{\ii}}\})
=\mathrm{e}^{y\mathcal{M}}\mathrm{e}^{N^2\sum_{i\ge 0}s_iu_{(i)}}.
\label{init_blp}
\end{align}
\end{theorem}

The partial differential equation (\ref{cut_join_bpl_ori}) agrees with 
the cut-and-equation obtained combinatorially in Theorem 1.1 of \cite{Andersen_new}.
Here we prove this theorem by rewriting the derivative
in the second term of the partial differential equation (\ref{blp_heat}), taking advantage of the following lemma.

\begin{lemma}\label{lem:M0_M2}
For a function $h(\{u_{\ii}\})$ of the generalized Miwa times $u_{\ii}$, the derivative in the second term of the partial differential equation (\ref{blp_heat}) can be rewritten as follows
\begin{align}
\frac{1}{2N}\mathrm{Tr}
\left[(\Lambda_{\mathrm{L}}^{-1})^{\mathrm{T}}\frac{\partial}{\partial\Lambda_{\mathrm{P}}}(\Lambda_{\mathrm{L}}^{-1})^{\mathrm{T}}\frac{\partial}{\partial\Lambda_{\mathrm{P}}}\right]h(\{u_{\ii}\})=
\left(M_0+\frac{1}{N^2}M_2\right)h(\{u_{\ii}\}).
\label{blp_op_h}
\end{align}
\end{lemma}
\begin{proof}
By the chain rule, the derivative on the left hand side of (\ref{blp_op_h}) is rewritten as follows
\begin{align}
&
{\mathrm{Tr}}\left[
(\Lambda_{\mathrm{L}}^{-1})^{\mathrm{T}}\frac{\partial}{\partial\Lambda_{\mathrm{P}}}(\Lambda_{\mathrm{L}}^{-1})^{\mathrm{T}}\frac{\partial}{\partial\Lambda_{\mathrm{P}}}\right]h(\{u_{\ii}\}) =
\nonumber \\
&
=\sum_{K\ge 0}\sum_{(i_1,\ldots,i_K)}
{\mathrm{Tr}}\left[
(\Lambda_{\mathrm{L}}^{-1})^{\mathrm{T}}\frac{\partial}{\partial\Lambda_{\mathrm{P}}}(\Lambda_{\mathrm{L}}^{-1})^{\mathrm{T}}\frac{\partial}{\partial\Lambda_{\mathrm{P}}}
u_{(i_1,\ldots,i_K)}\right]\frac{\partial}{\partial
 u_{(i_1,\ldots,i_K)}}h(\{u_{\ii}\})
\nonumber \\
&\quad
+\sum_{K,L\ge 0}\sum_{(i_1,\ldots,i_K)}\sum_{(j_1,\ldots,j_L)}
\mathrm{Tr}\left[
(\Lambda_{\mathrm{L}}^{-1})^{\mathrm{T}}\frac{\partial}{\partial\Lambda_{\mathrm{P}}}u_{(i_1,\ldots,i_K)}
(\Lambda_{\mathrm{L}}^{-1})^{\mathrm{T}}\frac{\partial}{\partial\Lambda_{\mathrm{P}}}
u_{(j_1,\ldots,j_L)}\right]
\nonumber \\
&\quad\quad\quad\quad\quad\quad\quad\quad\quad\quad\quad\quad
\times\frac{\partial^2}{\partial u_{(i_1,\ldots,i_K)}\partial u_{(j_1,\ldots,j_{L})}}h(\{u_{\ii}\}).
\nonumber
\end{align}
Each of the coefficients yields
\begin{align}
&
{\mathrm{Tr}}\left[
(\Lambda_{\mathrm{L}}^{-1})^{\mathrm{T}}\frac{\partial}{\partial\Lambda_{\mathrm{P}}}(\Lambda_{\mathrm{L}}^{-1})^{\mathrm{T}}\frac{\partial}{\partial\Lambda_{\mathrm{P}}}
u_{(i_1,\ldots,i_K)}\right] =
\nonumber \\
&
=\sum_{1\le I\ne M\le K}\sum_{\ell=0}^{i_I-1}\sum_{m=0}^{i_M-1}
\frac{1}{N}\mathrm{Tr}(
\Lambda_{\mathrm{P}}^{i_I-\ell-1}\Lambda_{\mathrm{L}}^{-1}\Lambda_{\mathrm{P}}^{i_{I+1}}\Lambda_{\mathrm{L}}^{-1}\cdots
\Lambda_{\mathrm{P}}^{i_{M-1}}\Lambda_{\mathrm{L}}^{-1}\Lambda_{\mathrm{P}}^{m}\Lambda_{\mathrm{L}}^{-1})
\nonumber \\
&\qquad\qquad\qquad\qquad\qquad
\times\mathrm{Tr}(
\Lambda_{\mathrm{P}}^{i_M-m-1}\Lambda_{\mathrm{L}}^{-1}\Lambda_{\mathrm{P}}^{i_{M+1}}\Lambda_{\mathrm{L}}^{-1}\cdots
\Lambda_{\mathrm{P}}^{i_{I-1}}\Lambda_{\mathrm{L}}^{-1}\Lambda_{\mathrm{P}}^{\ell}\Lambda_{\mathrm{L}}^{-1})
\nonumber \\
&\ \ \
+2\sum_{L=0}^K\sum_{\ell+m\le i_I-2}
\frac{1}{N}\mathrm{Tr}(\Lambda_{\mathrm{P}}^{\ell}\Lambda_{\mathrm{L}}^{-1}\Lambda_{\mathrm{P}}^{m}\Lambda_{\mathrm{L}}^{-1}\Lambda_{\mathrm{P}}^{i_{I+1}}\Lambda_{\mathrm{L}}^{-1}\cdots\Lambda_{\mathrm{P}}^{i_{L-1}}\Lambda_{\mathrm{L}}^{-1})
\mathrm{Tr}(\Lambda_{\mathrm{P}}^{i_I-\ell-m-2}\Lambda_{\mathrm{L}}^{-1})
\nonumber \\
&
=N\sum_{1\le I\ne M\le K}\sum_{\ell=0}^{i_I-1}\sum_{m=0}^{i_M-1}
u_{(i_I-\ell-1,i_{I+1},\ldots,i_{M-1},m)}
u_{(i_M-m-1,i_{M+1},\ldots,i_{I-1},\ell)}
\nonumber \\
&\ \ \
+2N\sum_{L=0}^K\sum_{\ell+m\le i_I-2}
u_{(\ell,m,i_{I+1},\ldots,i_{I-1})}u_{(i_I-\ell-m-2)},
\nonumber
\end{align}
and
\begin{align}
&
\mathrm{Tr}\left[
(\Lambda_{\mathrm{L}}^{-1})^{\mathrm{T}}\frac{\partial}{\partial\Lambda_{\mathrm{P}}}u_{(i_1,\ldots,i_K)}
(\Lambda_{\mathrm{L}}^{-1})^{\mathrm{T}}\frac{\partial}{\partial\Lambda_{\mathrm{P}}}
u_{(j_1,\ldots,j_L)}\right] =
\nonumber \\
&
=\sum_{I=1}^K\sum_{J=1}^L\sum_{\ell=0}^{i_I-1}\sum_{m=0}^{j_J-1}\frac{1}{N^2}
{\mathrm{Tr}}(\Lambda_{\mathrm{P}}^{i_I-\ell-1}\Lambda_{\mathrm{L}}^{-1}\Lambda_{\mathrm{P}}^{i_{I+1}}\Lambda_{\mathrm{L}}^{-1}\cdots\Lambda^{i_{I-1}}\Lambda_{\mathrm{L}}^{-1}\Lambda_{\mathrm{P}}^{\ell}\Lambda_{\mathrm{L}}^{-1}
\nonumber \\
&\qquad\qquad\qquad\qquad\qquad\quad\quad
\cdot
\Lambda_{\mathrm{P}}^{j_J-m-1}\Lambda_{\mathrm{L}}^{-1}\Lambda_{\mathrm{P}}^{j_{J+1}}\Lambda_{\mathrm{L}}^{-1}\cdots\Lambda_{\mathrm{P}}^{j_{J-1}}\Lambda_{\mathrm{L}}^{-1}\Lambda_{\mathrm{P}}^m\Lambda_{\mathrm{L}}^{-1})
\nonumber \\
&
=\frac{1}{N}\sum_{I=1}^K\sum_{J=1}^L\sum_{\ell=0}^{i_I-1}\sum_{m=0}^{j_J-1}
u_{(i_I-\ell-1,i_{I+1},\ldots,i_{I-1},\ell,j_J-m-1,j_{J+1},\ldots,j_{J-1},m)}.
\nonumber
\end{align}
In this way one obtains the right hand side of (\ref{blp_op_h}).
\end{proof}

As a corollary of Theorem \ref{thm:cut_join_blp}, one finds the cut-and-join equation for the 1-backbone generating function.\footnote{For $\Lambda_{\mathrm{L}}=I_N$ (or $s_i=s$ and $\Lambda_{\mathrm{P}}=0$) the cut-and-join equation for the 1-backbone generating function labeled by the 
boundary point spectrum (or boundary length spectrum) was proven combinatorially in \cite{AAPZ}.}
\begin{corollary}
The 1-backbone generating function $\mathcal{F}_1(x,y;\{s_i\};\{u_{\ii}\})$ obtained by picking up the $\mathcal{O}(s_i^{1})$ terms in $\mathcal{Z}_N(y;\{s_i\};\{u_{\ii}\})$ as follows
\begin{align}
\begin{split}
&
\mathcal{F}_1(N^{-1},y;\{s_i\};\{u_{\ii}\})
\\
&
=\frac{1}{{\mathrm{Vol}}_N}\int_{\mathcal{H}_N} dM\;\mathrm{e}^{-N\mathrm{Tr}\frac{M^2}{2}}
N\sum_{i\ge 0}s_i\mathrm{Tr}
(y^{1/2}\Lambda_{\mathrm{L}}^{-1}M+\Lambda_{\mathrm{P}})^i
\Lambda_{\mathrm{L}}^{-1},
\end{split}
\end{align}
obeys the cut-and-join equation
\begin{align}
\frac{\partial}{\partial y}\mathcal{F}_1(x,y;\{s_i\};\{u_{\ii}\})={\mathcal M}\mathcal{F}_1(x,y;\{s_i\};\{u_{\ii}\}),
\end{align}
where ${\mathcal M}=M_0+x^2M_2$. The solution is iteratively determined by
\begin{align}
\mathcal{F}_1(x,y;\{s_i\};\{u_{\ii}\})
=\mathrm{e}^{y\mathcal{M}}\mathcal{F}_1(x,y=0;\{s_i\};\{u_{\ii}\})
=\mathrm{e}^{y\mathcal{M}}\bigg(x^{-2}\sum_{i\ge 0}s_iu_{(i)}\bigg).
\end{align}
\end{corollary}

\section{\label{sec:non_orient}Non-oriented analogues}

In this section we consider the  enumeration of both orientable and non-orientable (jointly called non-oriented) partial chord diagrams \cite{AAPZ, Andersen_new}. To this end we generalize the matrix models introduced in Section \ref{sec:const_mat}. In Subsection \ref{subsec:non_or_mat}, matrix models for the boundary point spectrum, the boundary length spectrum, and the boundary length and point spectrum are introduced, based on the corresponding Gaussian matrix integrals over the space of rank $N$ real symmetric matrices. Subsequently, in Subsection \ref{subsec:non_ori_caj}, we derive cut-and-join equations for the generating functions of non-oriented partial chord diagrams, using analogous methods as those discussed in Section \ref{sec:mat_heat}.

\subsection{\label{subsec:non_or_mat}Non-oriented analogues of the matrix models}

In this subsection we generalize matrix models found in Section \ref{sec:const_mat}, in order to enumerate both orientable and non-orientable partial chord diagrams \cite{AAPZ, Andersen_new}.

\begin{definition}
Let $\widetilde{\mathcal N}_{h,k,l}(\{b_i\},\{n_i\},\{p_i\})$ denote the number of connected non-oriented partial chord diagrams of type $\{h, k, l; \{b_i\};\{n_i\};\{p_i\}\}$. Analogously as in the orientable case, we define
\begin{align}
&
\widetilde{\mathcal N}_{h,k,l}(\{b_i\},\{n_i\})
=\sum_{\{p_i\}}\widetilde{\mathcal N}_{h,k,l}(\{b_i\},\{n_i\},\{p_i\}),
\nonumber
\\
&
\widetilde{\mathcal N}_{h,k}(\{b_i\},\{p_i\})
=\sum_{\{n_i\}}\widetilde{\mathcal N}_{h,k,l=0}(\{b_i\},\{n_i\},\{p_i\}),
\nonumber
\end{align}
and introduce generating functions 
\begin{align}
\begin{split}
&
\widetilde{F}(x,y;\{s_i\};\{t_i\})
=\sum_{b\ge 1}\widetilde{F}_b(x,y;\{s_i\};\{t_i\}),
\\
&
\widetilde{F}_b(x,y;\{s_i\};\{t_i\})
=\frac{1}{b!}
\sum_{\sum_{i} b_i=b}\sum_{\{n_i\}}\widetilde{\mathcal N}_{h,k,l}(\{b_i\},\{n_i\})
x^{h-2}y^k
\prod_{i\ge 0}s_i^{b_i}t_i^{n_i},
\label{free_point_non}
\end{split}
\end{align}
and 
\begin{align}
\begin{split}
&
\widetilde{G}(x,y;\{s_i\};\{q_i\})
=\sum_{b\ge 1}\widetilde{G}_b(x,y;\{s_i\};\{q_i\}),
\\
&
\widetilde{G}_b(x,y;\{s_i\};\{q_i\})
=\frac{1}{b!}
\sum_{\sum_{i} b_i=b}\sum_{\{p_i\}}
\widetilde{\mathcal N}_{h,k}(\{b_i\},\{p_i\})x^{h-2}y^k
\prod_{i\ge 0}s_i^{b_i}\prod_{i\ge 1}q_i^{p_i}.
\label{free_length_non}
\end{split}
\end{align}
Generating functions of connected and disconnected partial chord diagrams are related by
\begin{align}
&
\widetilde{Z}^{\mathrm{P}}(x,y;\{s_i\};\{t_i\})=\exp\Big[\widetilde{F}(x,y;\{s_i\};\{t_i\})\Big],
\label{partition_point_non}
\\
&
\widetilde{Z}^{\mathrm{L}}(x,y;\{s_i\};\{q_i\})=\exp\Big[\widetilde{G}(x,y;\{s_i\};\{q_i\})\Big].
\label{partition_length_non}
\end{align}
\end{definition}


Furthermore, we introduce generating functions of non-oriented partial chord diagrams labeled by the boundary length and point spectrum.

\begin{definition}\label{def:blp_non}
Let $\widetilde{\mathcal N}_{h,k,l}(\{b_i\},\{n_{\ii}\})$ denote the 
number of connected orientable and non-orientable partial chord diagrams of type $\{h, k, l; \{b_i\};\{n_{\ii}\}\}$ with the boundary length and point spectrum $n_{\ii}$. We define the generating functions
\begin{align}
\begin{split}
&
\widetilde{\mathcal F}(x,y;\{s_i\};\{u_{\ii}\})=
\sum_{b\ge 1}\widetilde{\mathcal F}_b(x,y;\{s_i\};\{u_{\ii}\}),
\\
&
\widetilde{\mathcal F}_b(x,y;\{s_i\};\{u_{\ii}\})=\frac{1}{b!}
\sum_{\sum_ib_i=b}
\sum_{\{n_{\ii}\}}
\widetilde{\mathcal N}_{h,k,l}(\{b_i\},\{n_{\ii}\})x^{h-2}y^k
\prod_{i\ge 0}s_i^{b_i}\prod_{K\ge 1}\prod_{\{i_L\}_{L=1}^K} u_{\ii}^{n_{\ii}}.
\end{split}
\end{align}
As usual, generating functions of connected and disconnected partial chord diagrams are related by
\begin{align}
\widetilde{\mathcal Z}(x,y;\{s_i\};\{u_{\ii}\})=
\exp\Big[\widetilde{\mathcal F}(x,y;\{s_i\};\{u_{\ii}\})\Big].
\label{partition_bpl_non}
\end{align}
\end{definition}

\subsubsection*{\bfseries Non-oriented analogue of partial chord diagrams and Wick contractions}
A non-oriented partial chord diagram is a partial chord diagrams with each chord decorated by a binary variable, which indicates if it is twisted or not.
Such non-oriented partial chord diagrams are enumerated by real symmetric\footnote{
The Gaussian matrix integral over the space of real symmetric matrix is also referred to as the {\em Gaussian orthogonal ensemble} \cite{Dyson,Mehta}.
} matrix integrals. 
The Gaussian average $\langle{\mathcal O}(M)\rangle_N^{\mathrm{\widetilde{G}}}$ over the space ${\mathcal H}_N(\mathbb{R})$ of real symmetric matrices is defined by
\begin{align}
\langle{\mathcal O}(M)\rangle_N^{\mathrm{\widetilde{G}}}=\frac{1}{\mathrm{Vol}_N(\mathbb{R})}\int_{{\mathcal H}_N(\mathbb{R)}}dM\;{\mathcal O}(M)\,\mathrm{e}^{-N\mathrm{Tr}\frac{M^2}{4}},
\end{align}
where
\begin{align}
\mathrm{Vol}_N(\mathbb{R})=\int_{{\mathcal H}_N(\mathbb{R)}}dM\;\mathrm{e}^{-N\mathrm{Tr}\frac{M^2}{4}}
=N^{N(N+1)/2}\mathrm{Vol}({\mathcal H}_N(\mathbb{R})),
\label{volN_def_r}
\end{align}
For the choice of $\mathcal{O}(M)=M_{\alpha\beta}M_{\gamma\epsilon}$ ($\alpha,\beta,\gamma,\epsilon=1,\ldots,N$),
the Wick contraction is defined as
\begin{align}
\contraction[1ex]{}{X}{X}{}
M_{\alpha\beta}M_{\gamma\epsilon}
:=\langle M_{\alpha\beta}M_{\gamma\epsilon}\rangle_N^{\mathrm{\widetilde{G}}}
=\frac{1}{N}(\delta_{\alpha\epsilon}\delta_{\beta\gamma}+\delta_{\alpha\gamma}\delta_{\beta\epsilon}).
\label{Wick_contr_non}
\end{align} 
This Wick contraction consists of two terms, which encode the corresponding fatgraph as follows.
The first term $\frac{1}{N}\delta_{\alpha\epsilon}\delta_{\beta\gamma}$ is the same as in the Hermitian matrix integral (\ref{Wick_contr}), and it can be identified with an untwisted band in the two dimensional surface $\Sigma_c$ associated to the partial chord diagram $c$. 
The second term $\frac{1}{N}\delta_{\alpha\gamma}\delta_{\beta\epsilon}$ 
in (\ref{Wick_contr_non}) 
relates opposite matrix indices compared to the first term and 
 can be identified with the twisted band in $\Sigma_c$, see Figure \ref{chord_wick_non2}. 
Hence, for the real symmetric Gaussian average, the correspondence rules \textbf{C4, P5, L5} in Section \ref{sec:const_mat} are replaced by the following rules  \cite{GHJ,VWZ,Si,JNPI,MW,GM}.
\begin{description}
\item[N5] The Wick contraction between $M_{\alpha_j\beta_j}$ and $M_{\alpha'_{j'}\beta'_{j'}}$ corresponds to a band or a twisted band connecting two chord ends. 
Each Wick contraction imposes either the constraint $\delta_{\alpha_{j}\beta'_{j'}}\delta_{\alpha'_{j'}\beta_j}$ or the constraint $\delta_{\alpha_{j}\alpha'_{j'}}\delta_{\beta_{j}\beta'_{j'}}$ for matrix indices assigned to edges of chord ends matched by Wick contractions.
\end{description}

\begin{figure}[h]
\begin{center}
   \includegraphics[width=70mm,clip]{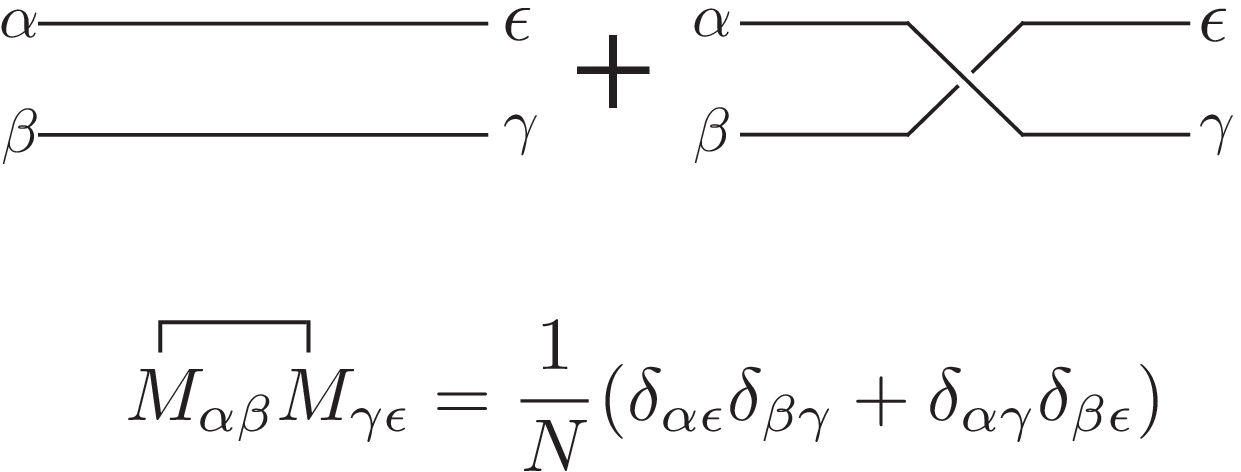}
\end{center}
\caption{\label{chord_wick_non2} Wick contraction and the untwisted / twisted bands.}
\end{figure}

In order to construct matrix models that enumerate non-oriented partial chord diagrams, we introduce two external real symmetric matrices matrices $\Omega_{\mathrm{P}}$ and $\Omega_{\mathrm{L}}$ 
\begin{align}
\Omega_{\mathrm{P}}=\Omega_{\mathrm{P}}^{\mathrm{T}},\quad 
\Omega_{\mathrm{L}}=\Omega_{\mathrm{L}}^{\mathrm{T}},
\nonumber
\end{align}
which take account of the fact that boundary cycles of non-oriented partial chord diagrams 
are not endowed with a specific orientation. To model the index structure $\ii$ of the boundary length and point spectrum correctly, we assume these two matrices do not commute
\begin{align}
[\Omega_{\mathrm{P}},\Omega_{\mathrm{L}}]\ne 0.
\nonumber
\end{align}
Furthermore, we introduce corresponding generalized Miwa times   
\begin{align}
&u_{(i_1,\ldots,i_K)}=\frac{1}{N}\mathrm{Tr}\left(
\Omega_{\mathrm{P}}^{i_1}\Omega_{\mathrm{L}}^{-1}\Omega_{\mathrm{P}}^{i_2}\Omega_{\mathrm{L}}^{-1}\cdots\Omega_{\mathrm{P}}^{i_K}\Omega_{\mathrm{L}}^{-1}
\right),
\label{blp_time_non}
\end{align}
which are invariant under the symmetry
\begin{align}
&u_{(i_1,\ldots,i_K)}=u_{(i_K,i_{K-1},\ldots,i_1)}.
\nonumber
\end{align}
This assignment implies the bijective correspondence (analogous to the orientable case discussed earlier) between non-oriented partial chord diagrams and Wick contractions, which is summarized in Table \ref{table:chord_wick_non}.

\begin{table}[htb]
\caption{\label{table:chord_wick_non} The correspondence between partial chord diagrams and operator products in the real symmetric matrix integral.}
  \begin{tabular}{|c|c|} \hline
    Partial chord diagram &   Gaussian average  \\ \hline \hline
    A chord end on a backbone & $\Omega_{\mathrm{L}}^{-1}M$ \\
    A marked point on a backbone & $\Omega_{\mathrm{P}}$ \\
    An underside of a backbone & $N\Omega_{\mathrm{L}}^{-1}$ \\
    A backbone & $N\mathrm{Tr}\left(
\Omega_{\mathrm{P}}^{\alpha_1}\Omega_{\mathrm{L}}^{-1}\Omega_{\mathrm{P}}^{\alpha_2}\Omega_{\mathrm{L}}^{-1}\cdots\Omega_{\mathrm{P}}^{\alpha_K}\Omega_{\mathrm{L}}^{-1}
\right)$ \\ 
    A Chord & A Wick contraction $\contraction[1ex]{}{X}{X}{}MM$ \\
\hline
  \end{tabular}
\end{table}
Using this correspondence,
generating functions $\widetilde{Z}^{\mathrm{P}}(x,y;\{s_i\};\{r_i\})$, $\widetilde{Z}^{\mathrm{L}}(x,y;\{s_i\};\{q_i\})$, and $\widetilde{\mathcal Z}(x,y;\{s_i\};\{u_{\ii}\})$
can be re-expressed in terms of matrix integrals. 
Repeating the same combinatorial arguments as for the orientable case in Section \ref{sec:const_mat}, we obtain the following three theorems.

\begin{theorem}\label{thm:matrix_point_non}
Let $\widetilde{Z}_N^{\mathrm{P}}(y;\{s_i\};\{r_i\})$ be the real symmetric matrix integral with the external symmetric matrix $\Omega_{\mathrm{P}}$ of rank $N$ 
\begin{align}
&
\widetilde{Z}_N^{\mathrm{P}}(y;\{s_i\};\{r_i\}) =
\nonumber \\
&
=\frac{1}{\mathrm{Vol}_N(\mathbb{R})}\int_{{\mathcal
 H}_N(\mathbb{R})} dM\;\exp\bigg[
-N\mathrm{Tr}\bigg(\frac{M^2}{4}-\sum_{i\ge 0}s_iy^{i/2}(M+y^{-1/2}\Omega_{\mathrm{P}})^i\bigg)
\bigg],
\label{partition_pt_non}
\end{align}
where $r_i$ are reverse Miwa times
\begin{align}
r_i=\frac{1}{N}\mathrm{Tr}\Omega_{\mathrm{P}}^i.
\label{reverse2}
\end{align}
This matrix integral agrees with the generating function (\ref{partition_point_non})
\begin{align}
\widetilde{Z}^{\mathrm{P}}_N(y;\{s_i\};\{r_i\})=\widetilde{Z}^{\mathrm{P}}(N^{-1},y;\{s_i\},t_0=1,\{t_i=r_i\}_{i\ge 1}).
\end{align}
The $t_0$-dependence can be implemented by the following rescaling of parameters
\begin{align}
\widetilde{Z}^{\mathrm{P}}_{t_0 N}(t_0y;\{t_0^{-1}s_i\};\{t_0^{-1}t_i\})=
\widetilde{Z}^{\mathrm{P}}(N^{-1},y;\{s_i\};\{t_i=r_i\}).
\end{align}
\end{theorem}

\begin{theorem}\label{thm:matrix_length_non}
Let $\widetilde{Z}_N^{\mathrm{L}}(y;\{s_i\};\{q_i\})$ be the real symmetric matrix integral with the external invertible symmetric matrix $\Omega_{\mathrm{L}}$ of rank $N$ 
\begin{align}
&
\widetilde{Z}_N^{\mathrm{L}}(y;\{s_i\};\{q_i\}) =
\nonumber \\
&
=\frac{1}{\mathrm{Vol}_N(\mathbb{R})}\int_{{\mathcal H}_N(\mathbb{R})}
dM\;\mathrm{exp}\bigg[
-N\mathrm{Tr}\bigg(
\frac{M^2}{4}-\sum_{i\ge 0}s_{i}y^{i/2}\left(\Omega_{\mathrm{L}}^{-1}M\right)^{i}\Omega_{\mathrm{L}}^{-1}
\bigg)
\bigg],
\label{Z_bdy_length3}
\end{align}
where $q_i$ are Miwa times
\begin{align}
q_i=\frac{1}{N}\mathrm{Tr}\Omega_{\mathrm{L}}^{-i}.
\label{Miwa2_non}
\end{align}
This matrix integral agrees with the generating function (\ref{partition_length_non})
\begin{align}
&
\widetilde{Z}^{\mathrm{L}}_N(y;\{s_i\};\{q_i\})=
\widetilde{Z}^{\mathrm{L}}(N^{-1},y;\{s_i\};\{q_i\}).
\end{align}
\end{theorem}

As considered in (\ref{Z_length}) and Subsection \ref{subsec:bl_h},  the specialization
$s_i=s$
of the matrix integral (\ref{Z_bdy_length3}) gives the following reduced model
\begin{align}
\widetilde{Z}_N^{\mathrm{L}}(y;s;\{q_i\})&=
\widetilde{Z}^{\mathrm{L}}_N(y;\{s_i=s\};\{q_i\}) =
\nonumber\\
&
=\frac{1}{\mathrm{Vol}_N({\IR})}\int_{{\mathcal H}_N({\IR})} dM\;\exp\bigg[
-N\mathrm{Tr}\bigg(\frac{M^2}{4}
+\frac{s}{y^{1/2}M-\Omega_{\mathrm{L}}}\bigg)
\bigg].
\label{Z_bdy_length3sp}
\end{align}
The cut-and-join equation that follows from this reduced model is derived in the next subsection.

\begin{theorem}\label{thm:blp_matrix_non}
Let $\widetilde{\mathcal Z}_N(y;\{s_i\};\{u_{\ii}\})$ be the real symmetric matrix integral with the external invertible symmetric matrices $\Omega_{\mathrm{P}}$ and $\Omega_{\mathrm{L}}$ of rank $N$ 
\begin{align}
&
\widetilde{\mathcal Z}_N(y;\{s_i\};\{u_{\ii}\}) =
\nonumber \\
&
=\frac{1}{\mathrm{Vol}_N(\mathbb{R})}\int_{{\mathcal H}_N(\mathbb{R})} dM\;\exp\bigg[
-N\mathrm{Tr}\bigg(\frac{M^2}{4}-\sum_{i\ge 0}s_i(y^{1/2}\Omega_{\mathrm{L}}^{-1}M+\Omega_{\mathrm{P}})^i\Omega_{\mathrm{L}}^{-1}\bigg)\bigg],
\label{blp_model_non}
\end{align}
and $u_{\ii}$ be the generalized Miwa times defined in (\ref{blp_time_non}).
This matrix integral agrees with the generating function (\ref{partition_bpl_non})
\begin{align}
\widetilde{\mathcal Z}_N(y;\{s_i\};\{u_{\ii}\})=
\widetilde{\mathcal Z}(N^{-1},y;\{s_i\};\{u_{\ii}\}).
\end{align}
\end{theorem}

\subsection{\label{subsec:non_ori_caj}Non-oriented analogues of cut-and-join equations}

We derive now non-oriented analogues of cut-and-join equations discussed in Section \ref{sec:mat_heat}. Analogously to the Hermitian matrix integral in (\ref{gblp_def}), 
we introduce the following matrix integral.

\begin{definition}\label{def:master_mat_r}
Let $U=U^{\mathrm{T}}$ and $V=V^{\mathrm{T}}$ be rank $N$ invertible symmetric matrices.
We define a formal real symmetric matrix integral with parameters $y$, $\{g_i\}_{i=-\infty}^{+\infty}$ as follows
\begin{align}
\begin{split}
&
\widetilde{Z}_N(y;\{g_i\};U;V) =
\\
&
=\frac{1}{\mathrm{Vol}_N({\IR})}\int_{\mathcal{H}_N({\IR})} dM\;
\exp\bigg[-N\mathrm{Tr}\bigg(\frac14M^2-\sum_{i\in {\IZ}}g_i(y^{1/2}V^{-1}M+U)^iV^{-1}\bigg)\bigg].
\label{o_gblp_def}
\end{split}
\end{align}
\end{definition}

The matrix integrals discussed in the previous subsection follow from this matrix integral by specializations
\begin{align}
&
\widetilde{Z}_N^{\mathrm{P}}(y;\{s_i\};\{r_i\}):\ \
g_{i<0}=0,\ \ \ g_{i\ge 0}=s_i,\ \ \
U=\Omega_{\mathrm{P}},\ \ \ V=I_N,
\label{sp_bp_no}
\\
&
\widetilde{Z}_N^{\mathrm{L}}(y;\{s_i\};\{q_i\}):\ \
g_{i<0}=0,\ \ \ g_{i\ge 0}=s_i,\ \ \
U=0,\ \ \ V=\Omega_{\mathrm{L}},
\label{sp_bl_no}
\\
&
\widetilde{Z}_N^{\mathrm{L}}(y;s;\{q_i\}):\ \
g_{i\neq -1}=0,\ \ \
g_{-1}=-s,\ \ \
U=\Omega_{\mathrm{L}},\ \ \ V=-I_N,
\label{sp_bl_sp_no}
\\
&
\widetilde{\mathcal{Z}}_N(y;\{s_i\};\{u_{\ii}\}):\ \
g_{i<0}=0,\ \ \ g_{i\ge 0}=s_i,\ \ \
U=\Omega_{\mathrm{P}},\ \ \ V=\Omega_{\mathrm{L}},
\label{sp_blp_no}
\end{align}
where $I_N$ is the rank $N$ identity matrix.

In Appendix \ref{app:heat_no} we prove the following proposition.

\begin{proposition}\label{prop:master_heat_r}
The matrix integral $\widetilde{Z}_N(y;\{g_i\};U;V)$ in (\ref{o_gblp_def}) obeys the partial differential equation
\begin{equation}
\bigg[\frac{\partial}{\partial y}
-\frac{1}{4N}\mathrm{Tr}(V^{-1})^{\mathrm{T}}\frac{\partial}{\partial A}
(V^{-1})^{\mathrm{T}}\frac{\partial}{\partial A}\bigg]\widetilde{Z}_N(y;\{g_i\};U;V)=0,
\label{o_gblp_heat}
\end{equation}
where $A$ is a matrix such that
\begin{align}
U=A+A^{\mathrm{T}}.
\nonumber
\end{align}
\end{proposition}

From this proposition and by the specializations (\ref{sp_bp_no}), (\ref{sp_bl_sp_no}), and (\ref{sp_blp_no})
we find partial differential equations for 
the corresponding matrix integrals.
For the specialization (\ref{sp_bl_no}), because of $U=0$ (and thus $A=0$),
the partial differential equation (\ref{o_gblp_heat}) cannot be reduced to a partial differential equation.

\begin{corollary}\label{cor:point_heat_no}
The matrix integral $\widetilde{Z}_N^{\mathrm{P}}(y;\{s_i\};\{r_i\})$ in (\ref{partition_pt_non}), $\widetilde{Z}_N^{\mathrm{L}}(y;s;\{q_i\})$ in (\ref{Z_bdy_length3sp}) and $\widetilde{\mathcal{Z}}_N(y;\{s_i\};\{u_{\ii}\})$ in (\ref{blp_model_non}) obey partial differential equations
\begin{align}
\begin{split}
&
\bigg[\frac{\partial}{\partial y}
-\frac{1}{4N}\mathrm{Tr}\frac{\partial^2}{\partial \Lambda_{\mathrm{P}}^2}
\bigg]\widetilde{Z}_N^{\mathrm{P}}(y;\{s_i\};\{r_i\})=0,
\\
&
\bigg[\frac{\partial}{\partial y}
-\frac{1}{4N}\mathrm{Tr}\frac{\partial^2}{\partial \Lambda_{\mathrm{L}}^2}
\bigg]\widetilde{Z}_N^{\mathrm{L}}(y;s;\{q_i\})=0,
\\
&
\bigg[\frac{\partial}{\partial y}
-\frac{1}{4N}\mathrm{Tr}
(\Omega_{\mathrm{L}}^{-1})^{\mathrm{T}}\frac{\partial}{\partial \Lambda_{\mathrm{P}}}
(\Omega_{\mathrm{L}}^{-1})^{\mathrm{T}}\frac{\partial}{\partial \Lambda_{\mathrm{P}}}\bigg]
\widetilde{\mathcal{Z}}_N(y;\{s_i\};\{u_{\ii}\})=0,
\end{split}
\end{align}
where $\Lambda_{\mathrm{P}}$ and $\Lambda_{\mathrm{L}}$ are matrices satisfying
\begin{align}
\Omega_{\mathrm{P}}=\Lambda_{\mathrm{P}}+\Lambda_{\mathrm{P}}^{\mathrm{T}},
\qquad
\Omega_{\mathrm{L}}=\Lambda_{\mathrm{L}}+\Lambda_{\mathrm{L}}^{\mathrm{T}}.
\nonumber
\end{align}
\end{corollary}

From this corollary we obtain non-oriented analogues of cut-and-join equations, by rewriting the derivatives with respect to the external matrices $\Lambda_{\mathrm{P}}$ and $\Lambda_{\mathrm{L}}$ in Corollary \ref{cor:point_heat_no} in terms of Miwa times $r_i$ in (\ref{reverse2}), 
$q_i$ in (\ref{Miwa2_non}), and $u_{\ii}$ in (\ref{blp_time_non}) as follows.

\begin{lemma}\label{lem:ext_Miwa_non}
Let $L_1$, $K_1$, $M_1$, and $M_2^{\vee}$ denote differential operators
\begin{align}
L_1&=\frac{1}{2}\sum_{i\ge 1}i(i+1)r_{i}\frac{\partial}{\partial
 r_{i+2}},
\label{L1_pt}\\
K_1&=\frac{1}{2}\sum_{i\ge 3}(i-2)(i-1)q_{i}\frac{\partial}{\partial
 q_{i-2}},
\label{K1}\\
\begin{split}
M_1&=\frac{1}{2}\sum_{K\ge
1}\sum_{\{i_1,\ldots,i_K\}}\sum_{1\le I\ne M\le
 K}\sum_{\ell=0}^{i_I-1}\sum_{m=0}^{i_M-1}
\\
&\qquad\quad
u_{(m,i_{M-1},i_{M-2},\ldots,i_{I+1},i_I-\ell-1
,i_M-m-1,i_{M+1},\ldots,i_{I-1},\ell)}
\frac{\partial}{\partial
u_{(i_1,\ldots,i_K)}}
\\
&\ \
+\sum_{K\ge 1}\sum_{\{i_1,\ldots,i_K\}}
\sum_{L=1}^K\sum_{\ell+m\le i_I-2}
u_{(\ell,i_I-\ell-m-2,m,i_{I+1},\ldots,i_{I-1})}
\frac{\partial}{\partial
u_{(i_1,\ldots,i_K)}},
\label{M1_non}
\end{split}
\end{align}
and
\begin{align}
\begin{split}
M_2^{\vee}
&=\frac{1}{2}\sum_{K,L\ge 1}\sum_{\{i_1,\ldots,i_K\}}
\sum_{\{j_1,\ldots,j_L\}}
\sum_{I=1}^K\sum_{J=1}^{L}\sum_{\ell=0}^{i_I-1}\sum_{m=0}^{j_J-1}
\\
&\ \ \ \
u_{(\ell,i_{I-1},\ldots,i_{I+1},i_{I}-\ell-1,j_J-m-1,j_{J+1},\ldots,j_{J-1},m)}
\frac{\partial^2}{\partial u_{(i_1,\ldots,i_K)}\partial u_{(j_1,\ldots,j_{L})}}.
\label{M2_non}
\end{split}
\end{align}
Then the derivatives with respect to $\Lambda_{\mathrm{P}}$ and $\Lambda_{\mathrm{L}}$ in Corollary \ref{cor:point_heat_no} are rewritten as
\begin{align}
\begin{split}
&
\frac{1}{4N}\mathrm{Tr}\frac{\partial^2}{\partial \Lambda_{\mathrm{P}}^2}f(\{r_i\})
=\left(L_0+\frac{1}{N}L_1+\frac{2}{N^2}L_2\right)f(\{r_i\}),
\\
&
\frac{1}{4N}\mathrm{Tr}\frac{\partial^2}{\partial \Lambda_{\mathrm{L}}^2}g(\{q_i\})
=\left(K_0+\frac{1}{N}K_1+\frac{2}{N^2}K_2\right)g(\{q_i\}),
\\
&
\frac{1}{4N}\mathrm{Tr}\left[
(\Omega_{\mathrm{L}}^{-1})^{\mathrm{T}}
\frac{\partial}{\partial\Lambda_{\mathrm{P}}}
(\Omega_{\mathrm{L}}^{-1})^{\mathrm{T}}
\frac{\partial}{\partial\Lambda_{\mathrm{P}}}\right]h(\{u_{\ii}\})
\\
&\qquad\qquad
=
\left(M_0+\frac{1}{N}M_1+\frac{1}{N^2}\big(M_2+M_2^{\vee}\big)\right)
h(\{u_{\ii}\}),
\end{split}
\end{align}
where $f(\{r_i\})$, $g(\{q_i\})$, and $h(\{u_{\ii}\})$ are functions of Miwa times $r_i$, $q_i$, and $u_{\ii}$, respectively. Here $L_{0,2}$, $K_{0,2}$, and $M_{0,2}$ are defined in (\ref{L02_pt}), (\ref{K02}), and (\ref{M02}), respectively.
\end{lemma}

The proof of this lemma is given in Appendix \ref{app:miwa_no}. By combining Corollary \ref{cor:point_heat_no} with Lemma \ref{lem:ext_Miwa_non}
one arrives at the following theorem.

\begin{theorem}\label{Thm:cut_join_non}
The matrix integrals $\widetilde{Z}_N^{\mathrm{P}}(y;\{s_i\};\{r_i\})$ in (\ref{partition_pt_non}), $\widetilde{Z}_N^{\mathrm{L}}(y;s;\{q_i\})$ in (\ref{Z_bdy_length3sp}), and $\widetilde{\mathcal{Z}}_N(y;\{s_i\};\{u_{\ii}\})$ in (\ref{blp_model_non}) obey the cut-and-join equations
\begin{align}
\begin{split}
&
\frac{\partial}{\partial y}\widetilde{Z}_N^{\mathrm{P}}(y;\{s_i\};\{r_i\})
=\widetilde{\mathcal L}\widetilde{Z}_N^{\mathrm{P}}(y;\{s_i\};\{r_i\}),
\\
&
\frac{\partial}{\partial y}\widetilde{Z}_N^{\mathrm{L}}(y;s;\{q_i\})
=\widetilde{\mathcal K}\widetilde{Z}_N^{\mathrm{L}}(y;s;\{q_i\}),
\\
&
\frac{\partial}{\partial y} \widetilde{\mathcal Z}_N(y;\{s_i\};\{u_{{\ii}}\})
=\widetilde{\mathcal M}\widetilde{\mathcal Z}_N(y;\{s_i\};\{u_{{\ii}}\}),
\label{cut_join_blp_non}
\end{split}
\end{align}
where
\begin{align}
&
\widetilde{\mathcal L}=L_0+\frac{1}{N}L_1+\frac{2}{N^2}L_2,
\nonumber
\\
&
\widetilde{\mathcal K}=K_0+\frac{1}{N}K_1+\frac{2}{N^2}K_2,
\nonumber
\\
&
\widetilde{\mathcal M}
=M_0+\frac{1}{N}M_1+\frac{1}{N^2}\big(M_2+M_2^{\vee}\big).
\nonumber
\end{align}
Assuming certain initial conditions at $y=0$, one can iteratively determine the above matrix integrals by solving the cut-and-join equations
\begin{align}
\begin{split}
&
\widetilde{Z}_N^{\mathrm{P}}(y;\{s_i\};\{r_i\})=\mathrm{e}^{y\widetilde{\mathcal L}}\widetilde{Z}_N^{\mathrm{P}}(y=0;\{s_i\};\{r_i\})
=\mathrm{e}^{y\widetilde{\mathcal L}}\mathrm{e}^{N^2\sum_{i\ge 0}s_ir_i},
\\
&
\widetilde{Z}_N^{\mathrm{L}}(y;s;\{q_i\})=\mathrm{e}^{y\widetilde{\mathcal K}}\widetilde{Z}_N^{\mathrm{L}}(y=0,s;\{q_i\})
=\mathrm{e}^{y\widetilde{\mathcal K}}\mathrm{e}^{N^2sq_1},
\\
&
\widetilde{\mathcal Z}_N(y;\{s_i\};\{u_{{\ii}}\})=
\mathrm{e}^{y\widetilde{\mathcal M}}\widetilde{\mathcal Z}_N(y=0;\{s_i\};\{u_{{\ii}}\})
=\mathrm{e}^{y\widetilde{\mathcal M}}\mathrm{e}^{-N^2\sum_{i\ge 0}s_iu_{(i)}}.
\end{split}
\end{align}
\end{theorem}

The cut-and-join equations (\ref{cut_join_blp_non}) agree with those of \cite{AAPZ, Andersen_new}. 
Finally, from Theorem \ref{Thm:cut_join_non} we find non-oriented analogues of cut-and-join equations for 1-backbone generating functions.

\begin{corollary}
The 1-backbone generating function $\widetilde{\mathcal{F}}_1(x,y;\{s_i\};\{u_{\ii}\})$ obtained by picking up the $\mathcal{O}(s_i^{1})$ term in $\widetilde{\mathcal{Z}}_N(y;\{s_i\};\{u_{\ii}\})$ is given by the following matrix integral
\begin{align}
\begin{split}
&
\widetilde{\mathcal{F}}_1(N^{-1},y;\{s_i\};\{u_{\ii}\}) =
\\
&
=\frac{1}{{\mathrm{Vol}}_N({\IR})}\int_{{\mathcal H}_N(\mathbb{R})} dM\;\mathrm{e}^{-N\mathrm{Tr}\frac{M^2}{4}}
N\sum_{i\ge 0}s_i\mathrm{Tr}
(y^{1/2}\Omega_{\mathrm{L}}^{-1}M+\Omega_{\mathrm{P}})^i
\Omega_{\mathrm{L}}^{-1},
\end{split}
\end{align}
and it obeys the cut-and-join equation
\begin{align}
\frac{\partial}{\partial y}\widetilde{\mathcal{F}}_1(x,y;\{s_i\};\{u_{\ii}\})
=\widetilde{{\mathcal M}}\widetilde{\mathcal{F}}_1(x,y;\{s_i\};\{u_{\ii}\}),
\end{align}
where $\widetilde{{\mathcal M}}=M_0+xM_1+x^2\big(M_2+M_2^{\vee}\big)$. The solution is iteratively determined by
\begin{align}
\widetilde{\mathcal{F}}_1(x,y;\{s_i\};\{u_{\ii}\})
=\mathrm{e}^{y\widetilde{\mathcal{M}}}\widetilde{\mathcal{F}}_1(x,y=0;\{s_i\};\{u_{\ii}\})
=\mathrm{e}^{y\widetilde{\mathcal{M}}}\bigg(x^{-2}\sum_{i\ge 0}s_iu_{(i)}\bigg).
\end{align}
\end{corollary}


\newpage
\appendix
\section{\label{app:heat_no}Proof of Proposition \ref{prop:master_heat_r}}

In this appendix we prove the Proposition \ref{prop:master_heat_r}, which states that the matrix integral
\begin{align}
\begin{split}
&
\widetilde{Z}_N(y;\{g_i\};U;V) =
\\
&
=\frac{1}{\mathrm{Vol}_N({\IR})}\int_{\mathcal{H}_N({\IR})} dM\;
\exp\bigg[-N\mathrm{Tr}\bigg(\frac14M^2-\sum_{i\in {\IZ}}g_i(y^{1/2}V^{-1}M+U)^iV^{-1}\bigg)\bigg],
\nonumber
\end{split}
\end{align}
obeys the partial differential equation
\begin{equation}
\bigg[\frac{\partial}{\partial y}
-\frac{1}{4N}\mathrm{Tr}(V^{-1})^{\mathrm{T}}\frac{\partial}{\partial A}
(V^{-1})^{\mathrm{T}}\frac{\partial}{\partial A}\bigg]\widetilde{Z}_N(y;\{g_i\};U;V)=0,
\label{o_gblp_heat_ap}
\end{equation}
where $A$ is a matrix that satisfies $U=A+A^{\mathrm{T}}$.

\begin{proof}
In order to differentiate the matrix integral $\widetilde{Z}_N(y;\{g_i\};U;V)$ with respect to $A$ we use the identities
\begin{align}
\begin{split}
&
\frac{\partial U_{\alpha\beta}}{\partial A_{\gamma\epsilon}}=\delta_{\alpha\gamma}\delta_{\beta\epsilon}+\delta_{\alpha\epsilon}\delta_{\beta\gamma},\\
&
\frac{\partial (y^{1/2}V^{-1}X)^{-1}_{\alpha\beta}}{\partial A_{\gamma\epsilon}}=
-(y^{1/2}V^{-1}X)^{-1}_{\alpha\gamma}(y^{1/2}V^{-1}X)^{-1}_{\beta\epsilon}
-(y^{1/2}V^{-1}X)^{-1}_{\alpha\epsilon}(y^{1/2}V^{-1}X)^{-1}_{\beta\gamma},
\nonumber
\end{split}
\end{align}
where $X=M+y^{-1/2}VU$. Using this shifted variable $X$ one obtains
\begin{align}
\frac{1}{2N}\frac{\partial}{\partial A_{\alpha\beta}}\widetilde{Z}_N(y;\{g_i\};U;V)&=
\bigg<\sum_{i=0}^{\infty}y^{(i-1)/2}g_i\sum_{j=0}^{i-1}\big((V^{-1}X)^jV^{-1}(V^{-1}X)^{i-j-1}\big)_{\alpha\beta}
\nonumber\\
&
-\sum_{i=1}^{\infty}y^{-(i+1)/2}g_{-i}\sum_{j=0}^{i-1}\big((V^{-1}X)^{-j-1}V^{-1}(V^{-1}X)^{-i+j}\big)_{\alpha\beta}\bigg>_{\IR},
\nonumber
\end{align}
where $\left<\cdots\right>_{\IR}$ denotes the unnormalized average
\begin{equation}
\left<\mathcal{O}(X)\right>_{\IR}=
\int_{\widetilde{\mathcal{H}}_N({\IR})} dX\; \mathcal{O}(X)
\exp\bigg[-N\mathrm{Tr}\bigg(\frac14(X-y^{-1/2}VU)^2-\sum_{i\in {\IZ}}y^{i/2}g_i(V^{-1}X)^iV^{-1}\bigg)\bigg].
\nonumber
\end{equation}
Here $\widetilde{\mathcal{H}}_N({\IR})$ is the space of shifted matrices $X=M+y^{-1/2}VU$ with $M \in \mathcal{H}_N({\IR})$. It follows that
\begin{align}
&
\frac{1}{2N^2}\mathrm{Tr}(V^{-1})^{\mathrm{T}}\frac{\partial}{\partial A}
(V^{-1})^{\mathrm{T}}\frac{\partial}{\partial A}\widetilde{Z}_N(y;\{g_i\};U;V) =
\nonumber\\
&
=\bigg<\sum_{i=0}^{\infty}y^{-1/2}g_i\sum_{j=0}^{i-1}
\mathrm{Tr}(X-y^{-1/2}VU)V^{-1}(y^{1/2}V^{-1}X)^jV^{-1}(y^{1/2}V^{-1}X)^{i-j-1}
\nonumber\\
&\ \ 
-\sum_{i=1}^{\infty}y^{-1/2}g_{-i}\sum_{j=0}^{i-1}
\mathrm{Tr}(X-y^{-1/2}VU)V^{-1}(y^{1/2}V^{-1}X)^{-j-1}V^{-1}(y^{1/2}V^{-1}X)^{-i+j}\bigg>_{\IR}.
\label{o_vava_p}
\end{align}
On the other hand, by differentiating the matrix integral 
$\widetilde{Z}_N(y;\{g_i\};U;V)$ with respect to $y$ one obtains the same expression as (\ref{o_vava_p}) times $N/2$, from which the partial differential equation (\ref{o_gblp_heat_ap}) is obtained.
\end{proof}

\section{\label{app:miwa_no}Proof of Lemma \ref{lem:ext_Miwa_non}}

In this appendix we prove the Lemma \ref{lem:ext_Miwa_non}, which states that for functions $f(\{r_i\})$, $g(\{q_i\})$, and $h(\{u_{\ii}\})$ of Miwa times $r_i$ in (\ref{reverse2}), $q_i$ in (\ref{Miwa2_non}), and $u_{\ii}$ in $(\ref{blp_time_non})$, we find
\begin{align}
&
\frac{1}{4N}\mathrm{Tr}\frac{\partial^2}{\partial \Lambda_{\mathrm{P}}^2}f(\{r_i\})
=\left(L_0+\frac{1}{N}L_1+\frac{2}{N^2}L_2\right)f(\{r_i\}),
\label{app_m_1}
\\
&
\frac{1}{4N}\mathrm{Tr}\frac{\partial^2}{\partial \Lambda_{\mathrm{L}}^2}g(\{q_i\})
=\left(K_0+\frac{1}{N}K_1+\frac{2}{N^2}K_2\right)g(\{q_i\}),
\label{app_m_2}
\\
\begin{split}
&
\frac{1}{4N}\mathrm{Tr}\left[
(\Omega_{\mathrm{L}}^{-1})^{\mathrm{T}}
\frac{\partial}{\partial\Lambda_{\mathrm{P}}}
(\Omega_{\mathrm{L}}^{-1})^{\mathrm{T}}
\frac{\partial}{\partial\Lambda_{\mathrm{P}}}\right]h(\{u_{\ii}\})
\\
&\qquad\qquad
=
\left(M_0+\frac{1}{N}M_1+\frac{1}{N^2}\big(M_2+M_2^{\vee}\big)\right)
h(\{u_{\ii}\}),
\label{app_m_3}
\end{split}
\end{align}
where $L_{0,1,2}$, $K_{0,1,2}$, $M_{0,1,2}$, and $M_2^{\vee}$ are defined in (\ref{L02_pt}), (\ref{K02}), (\ref{M02}), and in Lemma \ref{lem:ext_Miwa_non}. Here the matrices $\Lambda_{\mathrm{P}}$ and $\Lambda_{\mathrm{L}}$ satisfy $\Omega_{\mathrm{P}}=\Lambda_{\mathrm{P}}+\Lambda_{\mathrm{P}}^{\mathrm{T}}$ and
$\Omega_{\mathrm{L}}=\Lambda_{\mathrm{L}}+\Lambda_{\mathrm{L}}^{\mathrm{T}}$.

\begin{proof}
First we prove (\ref{app_m_1}). Consider the derivative $\partial/\partial\Lambda_{\mathrm{P}}$ of the reverse Miwa time $r_i$
\begin{align}
\frac{\partial r_i}{\partial \Lambda_{\mathrm{P}\beta\alpha}}
=\frac{2i}{N}\Omega_{\mathrm{P}\alpha\beta}^{i-1}, \qquad
\mathrm{Tr}\frac{\partial^2r_i}{\partial\Lambda_{\mathrm{P}}^2}
=2Ni\sum_{j=1}^{i-1}r_{j-1}r_{i-j-1}+2i(i-1)r_{i-2}.
\nonumber
\end{align}
Using these relations the left hand side of (\ref{app_m_1}) is rewritten as
\begin{align}
\begin{split}
\frac{1}{4N}\mathrm{Tr}\frac{\partial^2}{\partial\Lambda_{\mathrm{P}}^2}f(\{r_i\})
&=
\frac{1}{4N}\sum_{i\ge 1}\mathrm{Tr}\frac{\partial^2 r_i}{\partial
\Lambda_{\mathrm{P}}^2}\frac{\partial f(\{r_i\})}{\partial r_i}
+
\frac{1}{4N}\sum_{i,j\ge 1}\mathrm{Tr}
\frac{\partial r_i}{\partial \Lambda_{\mathrm{P}}}\frac{\partial r_j}{\partial
\Lambda_{\mathrm{P}}}
\frac{\partial^2f(\{r_i\})}{\partial r_i\partial r_j}
\\
&=
\frac{1}{2}\sum_{i\ge
2}\sum_{j=1}^{i-1}ir_{j-1}r_{i-j-1}\frac{\partial f(\{r_i\})}{\partial r_i}
+\frac{1}{2N}\sum_{i\ge 2}i(i-1)r_{i-2}\frac{\partial
f(\{r_i\})}{\partial r_i}
\\
&\ \ \
+\frac{1}{N^2}\sum_{i,j\ge 1}ijr_{i+j-2}\frac{\partial}{\partial r_i}\frac{\partial}{\partial r_j}f(\{r_i\}).
\nonumber
\end{split}
\end{align}
This agrees with the right hand side of (\ref{app_m_1}).

Second, we prove (\ref{app_m_2}). Using the identity
\begin{align}
\frac{\partial \Omega^{-1}_{\mathrm{L}\gamma\epsilon}}{\partial \Lambda_{\mathrm{L}\alpha\beta}}
=-\Omega^{-1}_{\mathrm{L}\beta\epsilon}\Omega^{-1}_{\mathrm{L}\gamma\alpha}-\Omega^{-1}_{\mathrm{L}\alpha\epsilon}\Omega^{-1}_{\mathrm{L}\gamma\beta},
\nonumber 
\end{align}
one finds that
\begin{align}
\frac{\partial q_i}{\partial
\Lambda_{\mathrm{L}\alpha\beta}}=-2\frac{i}{N}\Omega^{-i-1}_{\mathrm{L}\beta\alpha},
\qquad
\mathrm{Tr}\frac{\partial^2}{\partial \Lambda_{\mathrm{L}}^2}q_i
=2i(i+1)q_{i+2}+2iN\sum_{j=0}^iq_{i+1}q_{j+1}.
\nonumber
\end{align}
Then the left hand side of (\ref{app_m_2}) yields
\begin{align}
\frac{1}{4N}\mathrm{Tr}\frac{\partial^2}{\partial\Lambda_{\mathrm{L}}^2}g(\{q_i\})
&=
\frac{1}{2}\sum_{i\ge 1}\sum_{j=0}^iiq_{i+1}q_{j+1}\frac{\partial
g(\{q_i\})}{\partial q_i}
+
\frac{1}{2N}\sum_{i\ge 1}i(i+1)q_{i+2}\frac{\partial g(\{q_i\})}{\partial q_i}
\nonumber\\
&\ \ \
+\frac{1}{N^2}\sum_{i,j\ge 1}ijq_{i+j+2}
\frac{\partial^2g(\{q_i\})}{\partial q_i\partial q_j}.
\nonumber
\end{align}
This agrees with the right hand side of (\ref{app_m_2}).

Finally we prove (\ref{app_m_3}). Using the chain rule, the derivative action on the left hand side of (\ref{app_m_3}) is written as
\begin{align}
&
{\mathrm{Tr}}\left[
(\Omega_{\mathrm{L}}^{-1})^{\mathrm{T}}\frac{\partial}{\partial\Lambda_{\mathrm{P}}}(\Omega_{\mathrm{L}}^{-1})^{\mathrm{T}}\frac{\partial}{\partial\Lambda_{\mathrm{P}}}\right]h(\{u_{\ii}\}) =
\nonumber \\
&
=\sum_{K\ge 1}\sum_{\{i_1,\ldots,i_K\}}
{\mathrm{Tr}}\left[
(\Omega_{\mathrm{L}}^{-1})^{\mathrm{T}}\frac{\partial}{\partial\Lambda_{\mathrm{P}}}(\Omega_{\mathrm{L}}^{-1})^{\mathrm{T}}\frac{\partial}{\partial\Lambda_{\mathrm{P}}}
u_{(i_1,\ldots,i_K)}\right]\frac{\partial}{\partial
 u_{(i_1,\ldots,i_K)}}h(\{u_{\ii}\})
\nonumber \\
&\quad
+\sum_{K,L\ge 1}\sum_{\{i_1,\ldots,i_K\}}\sum_{\{j_1,\ldots,j_L\}}
\mathrm{Tr}\left[
(\Omega_{\mathrm{L}}^{-1})^{\mathrm{T}}\frac{\partial}{\partial\Lambda_{\mathrm{P}}}u_{(i_1,\ldots,i_K)}
(\Omega_{\mathrm{L}}^{-1})^{\mathrm{T}}\frac{\partial}{\partial\Lambda_{\mathrm{P}}}
u_{(j_1,\ldots,j_L)}\right]
\nonumber \\
&\quad\quad\quad\quad\quad\quad\quad\quad\quad\quad\quad\quad
\times\frac{\partial^2}{\partial u_{(i_1,\ldots,i_K)}\partial u_{(j_1,\ldots,j_{L})}}h(\{u_{\ii}\}).
\nonumber
\end{align}
Each of the coefficients yields
\begin{align}
&
{\mathrm{Tr}}\left[
(\Omega_{\mathrm{L}}^{-1})^{\mathrm{T}}\frac{\partial}{\partial\Lambda_{\mathrm{P}}}(\Omega_{\mathrm{L}}^{-1})^{\mathrm{T}}\frac{\partial}{\partial\Lambda_{\mathrm{P}}}
u_{(i_1,\ldots,i_K)}\right] =
\nonumber \\
&
=2\sum_{1\le I\ne M\le K}^K\sum_{\ell=0}^{i_I-1}\sum_{m=0}^{i_M-1}
\frac{1}{N}\Bigl(
\mathrm{Tr}(
\Omega_{\mathrm{P}}^{i_I-\ell-1}\Omega_{\mathrm{L}}^{-1}\Omega_{\mathrm{P}}^{i_{I+1}}\Omega_{\mathrm{L}}^{-1}\cdots
\Omega_{\mathrm{P}}^{i_{M-1}}\Omega_{\mathrm{L}}^{-1}\Omega_{\mathrm{P}}^{m}\Omega_{\mathrm{L}}^{-1})
\nonumber \\
&\qquad\qquad\qquad\qquad\qquad\qquad
\times\mathrm{Tr}(
\Omega_{\mathrm{P}}^{i_M-m-1}\Omega_{\mathrm{L}}^{-1}\Omega_{\mathrm{P}}^{i_{M+1}}\Omega_{\mathrm{L}}^{-1}\cdots
\Omega_{\mathrm{P}}^{i_{I-1}}\Omega_{\mathrm{L}}^{-1}\Omega_{\mathrm{P}}^{\ell}\Omega_{\mathrm{L}}^{-1})
\nonumber \\
&\qquad\qquad\qquad\qquad\qquad
+\mathrm{Tr}(\Omega_{\mathrm{P}}^{m}\Omega_{\mathrm{L}}^{-1}\Omega_{\mathrm{P}}^{i_{M-1}}\Omega_{\mathrm{L}}^{-1}\Omega_{\mathrm{P}}^{i_{M-2}}\Omega_{\mathrm{L}}^{-1}\cdots
\Omega_{\mathrm{P}}^{i_{I+1}}\Omega_{\mathrm{L}}^{-1}\Omega_{\mathrm{P}}^{i_I-\ell-1}\Omega_{\mathrm{L}}^{-1})
\nonumber \\
&\qquad\qquad\qquad\qquad\qquad\qquad\quad
\cdot
\Omega_{\mathrm{P}}^{i_M-m-1}\Omega_{\mathrm{L}}^{-1}\Omega_{\mathrm{P}}^{i_{M+1}}\Omega_{\mathrm{L}}^{-1}\Omega_{\mathrm{P}}^{i_{M+2}}\Omega_{\mathrm{L}}^{-1}\cdots
\Omega_{\mathrm{P}}^{i_{I-1}}\Omega_{\mathrm{L}}^{-1}\Omega_{\mathrm{P}}^{\ell}\Omega_{\mathrm{L}}^{-1})
\Bigr)
\nonumber \\
&\ \ \
+4\sum_{I=1}^K\sum_{\ell+m\le i_I-2}\Bigl(
\frac{1}{N}
\mathrm{Tr}(\Omega_{\mathrm{P}}^{\ell}\Omega_{\mathrm{L}}^{-1}\Omega_{\mathrm{P}}^{m}\Omega_{\mathrm{L}}^{-1}\Omega_{\mathrm{P}}^{i_{I+1}}\Omega_{\mathrm{L}}^{-1}\cdots\Omega_{\mathrm{P}}^{i_{I-1}}\Omega_{\mathrm{L}}^{-1})
\mathrm{Tr}(\Omega_{\mathrm{P}}^{i_I-\ell-m-2}\Omega_{\mathrm{L}}^{-1})
\nonumber \\
&\qquad\qquad\qquad\qquad\quad
+\mathrm{Tr}(\Omega_{\mathrm{P}}^{\ell}\Omega_{\mathrm{L}}^{-1}\Omega_{\mathrm{P}}^{m}\Omega_{\mathrm{L}}^{-1}\Omega_{\mathrm{P}}^{i_I-\ell-m-2}\Omega_{\mathrm{L}}^{-1}\Omega_{\mathrm{P}}^{i_{I+1}}\Omega_{\mathrm{L}}^{-1}\cdots\Omega_{\mathrm{P}}^{i_{I-1}}\Omega_{\mathrm{L}}^{-1})\Bigr)
\nonumber \\
&
=2\sum_{1\le I\ne M\le K}
\sum_{\ell=0}^{i_I-1}\sum_{m=0}^{i_M-1}
\Big(Nu_{(i_I-\ell-1,i_{I+1},\ldots,i_{M-1},m)}
u_{(i_M-m-1,i_{M+1},\ldots,i_{I-1},\ell)}
\nonumber \\
&\qquad\qquad\qquad\qquad\qquad\quad
+u_{(m,i_{M-1},i_{M-2}\ldots,i_{I+1},i_I-\ell-1,
i_M-m-1,i_{M+1},\ldots,i_{I-1},\ell)}\Big)
\nonumber \\
&\ \ \
+4\sum_{I=0}^K\sum_{\ell+m\le i_I-2}
\Big(Nu_{(\ell,m,i_{I+1},\ldots,i_{I-1})}u_{(i_I-\ell-m-2)}
+u_{(\ell,i_I-\ell-m-2,m,i_{I+1},\ldots,i_{I-1})}\Big),
\nonumber 
\end{align}
and
\begin{align}
&
\mathrm{Tr}\left[
(\Omega_{\mathrm{L}}^{-1})^{\mathrm{T}}\frac{\partial}{\partial\Lambda_{\mathrm{P}}}u_{(i_1,\ldots,i_K)}
(\Omega_{\mathrm{L}}^{-1})^{\mathrm{T}}\frac{\partial}{\partial\Lambda_{\mathrm{P}}}
u_{(j_1,\ldots,j_L)}\right] =
\nonumber \\
&
=\sum_{I=1}^K\sum_{J=1}^L\sum_{\ell=0}^{i_I-1}\sum_{m=0}^{j_J-1}
\frac{2}{N^2}\biggl(
{\mathrm{Tr}}(\Omega_{\mathrm{P}}^{i_I-\ell-1}\Omega_{\mathrm{L}}^{-1}\Omega_{\mathrm{P}}^{i_{I+1}}\Omega_{\mathrm{L}}^{-1}\cdots\Omega^{i_{I-1}}\Omega_{\mathrm{L}}^{-1}\Omega_{\mathrm{P}}^{\ell}\Omega_{\mathrm{L}}^{-1}
\nonumber \\
&\qquad\qquad\qquad\qquad\qquad\qquad
\cdot
\Omega_{\mathrm{P}}^{j_J-m-1}\Omega_{\mathrm{L}}^{-1}\Omega_{\mathrm{P}}^{j_{J+1}}\Omega_{\mathrm{L}}^{-1}\cdots\Omega_{\mathrm{P}}^{j_{J-1}}\Omega_{\mathrm{L}}^{-1}\Omega_{\mathrm{P}}^m\Omega_{\mathrm{L}}^{-1})
\nonumber \\
&\qquad\qquad\qquad\qquad\qquad\quad
+
{\mathrm{Tr}}(\Omega_{\mathrm{P}}^{\ell}\Omega_{\mathrm{L}}^{-1}\Omega_{\mathrm{P}}^{i_{I-1}}\Omega_{\mathrm{L}}^{-1}\cdots\Omega^{i_{I+1}}\Omega_{\mathrm{L}}^{-1}\Omega_{\mathrm{P}}^{i_I-\ell-1}\Omega_{\mathrm{L}}^{-1}
\nonumber \\
&\qquad\qquad\qquad\qquad\qquad\qquad\qquad
\cdot
\Omega_{\mathrm{P}}^{j_J-m-1}\Omega_{\mathrm{L}}^{-1}\Omega_{\mathrm{P}}^{j_{J+1}}\Omega_{\mathrm{L}}^{-1}\cdots\Omega_{\mathrm{P}}^{j_{J-1}}\Omega_{\mathrm{L}}^{-1}\Omega_{\mathrm{P}}^m\Omega_{\mathrm{L}}^{-1})
\biggr)
\nonumber \\
&
=\frac{2}{N}\sum_{I=1}^K\sum_{J=1}^L\sum_{\ell=0}^{i_I-1}\sum_{m=0}^{j_J-1}
\biggl(
u_{(i_I-\ell-1,i_{I+1},\ldots,i_{I-1},\ell,j_J-m-1,j_{J+1},\ldots,j_{J-1},m)}
\nonumber \\
&\qquad\qquad\qquad\qquad\qquad\quad
+u_{(\ell,i_{I-1},\ldots,i_{I+1},i_I-\ell-1,j_J-m-1,j_{J+1},\ldots,j_{J-1},m)}
\biggr).
\nonumber
\end{align}
Then one obtains the right hand side of (\ref{app_m_3}).
\end{proof}


\end{document}